\documentclass[a4paper,11pt]{article}
\usepackage{latexsym,amsfonts,amsthm,epsf}
\newtheorem{lemma}{Lemma}[section]
\newtheorem{theorem}[lemma]{Theorem}
\newtheorem{corollary}[lemma]{Corollary}
\newtheorem{definition}[lemma]{Definition}
\newtheorem{proposition}[lemma]{Proposition}
\newtheorem{claim}[lemma]{Claim}
\newtheorem{observation}[lemma]{Observation}

\newcommand{\bahar}[1]{\textcolor{black}{#1}}
\newcommand{\david}[1]{\textcolor{black}{#1}}
\newcommand{\piotr}[1]{\textcolor{black}{#1}}

\usepackage{tkz-graph}
\usetikzlibrary{calc}
\usepackage[ruled,vlined]{algorithm2e}
\usepackage{url}
\usepackage{helvet}
\usepackage{courier}
\usepackage{graphicx}
\usepackage{amssymb}
\usepackage{dsfont}
\usepackage{amsmath}
\usepackage{color}
\usepackage[textsize=tiny]{todonotes}

\newcommand{\todoB}[1]{\todo[color=pink]{B: #1}{}}

\usepackage[normalem]{ulem}

\newcommand{\agentset}{N}
\newcommand{\machineset}{O}
\newcommand{\objectset}{\machineset}
\newcommand{\agent}{i}
\newcommand{\agentone}{1}
\newcommand{\agenttwo}{2}
\newcommand{\agentthree}{3}
\newcommand{\object}{o}

\newcommand{\inputsetting}{I}
\newcommand{\inputsettingset}{\mathcal{I}}

\newcommand{\weaklypref}{\succeq}
\newcommand{\strictlypref}{\succ}
\newcommand{\tie}{\simeq}
\newcommand{\rank}{rank}
\newcommand{\matching}{\mu}
\newcommand{\eqclass}{C}

\newcommand{\mechanism}{\phi}
\newcommand{\agentsorder}{\sigma}
\newcommand{\sig}{\rho}
\newcommand{\numagents}{n_1}
\newcommand{\numobjects}{n_2}

\newcommand{\numvertices}{n}
\newcommand{\preflist}{L}
\newcommand{\preflistset}{\mathcal{L}}
\newcommand{\weight}{w}
\newcommand{\weightprofile}{W}
\newcommand{\randomlist}{R} 
\newcommand{\matchingset}{\mathcal{M}} 

\renewcommand{\Pr}[1]{\mbox{\rm\bf Pr}\left[#1\right]}

\newcommand{\note}[1]{~\\\frame{\begin{minipage}[c]{\textwidth}\vspace{2pt}{#1}\vspace{2pt}\end{minipage}}\vspace{3pt}\\}
\providecommand{\note}[1]{}

\newcommand{\hide}[1]{}

\newcommand{\E}{\mathbb{E}}

\begin{document}
\title{\bf Size versus truthfulness in the \\ House Allocation problem\thanks{Supported by grants  EP/K01000X/1, EP/K010042/1 and EP/028306/01 from the Engineering and Physical Sciences Research Council. A preliminary version of this paper was published in the proceedings of the 15th ACM Conference on Economics and Computation (EC 2014).}}
\author{Piotr Krysta$^1$, David Manlove$^2$, Baharak Rastegari$^3$ and Jinshan Zhang$^1$\\
\\ 
\small
\emph{$^1$ Department of Computer Science, University of Liverpool,} \\
\small
\emph{Ashton Building, Ashton Street, Liverpool L69 3BX, UK.}
\\
\small
\emph{Email: {\tt p.krysta@liverpool.ac.uk}.} \\
\\
\small
\emph{$^2$ School of Computing Science, University of Glasgow,} \\ 
\small
\emph{Sir Alwyn Williams Building, Glasgow G12 8QQ, Glasgow, UK.}
\\
\small
\emph{Email: {\tt david.manlove@glasgow.ac.uk}.} \\
\\
\small
\emph{$^3$ School of Electronics and Computer Science, University of Southampton} \\ 
\small
\emph{Email: {\tt b.rastegari@soton.ac.uk}.}}
\date{ }
\maketitle

\begin{abstract}
We study the House Allocation problem (also known as the Assignment problem), i.e., the problem of allocating a set of objects among a set of agents, where each agent has ordinal preferences (possibly involving ties) over a subset of the objects.  We focus on truthful mechanisms without monetary transfers for finding large Pareto optimal matchings.  It is straightforward to show that no deterministic truthful mechanism can approximate a maximum cardinality Pareto optimal matching with ratio better than 2.  We thus consider randomised mechanisms.  We give a natural and explicit extension of the classical Random Serial Dictatorship Mechanism (RSDM) specifically for the House Allocation problem where preference lists can include ties.  We thus obtain a universally truthful randomised mechanism for finding a Pareto optimal matching and show that it achieves an approximation ratio of $\frac{e}{e-1}$.  The same bound holds even when agents have priorities (weights) and our goal is to find a maximum weight (as opposed to maximum cardinality) Pareto optimal matching.  On the other hand we give a lower bound of $\frac{18}{13}$ on the approximation ratio of any universally truthful Pareto optimal mechanism in settings with strict preferences.  In the case that the mechanism must additionally be non-bossy \piotr{with an additional technical assumption}, we show by \david{utilizing a result of Bade} that an improved lower bound of $\frac{e}{e-1}$ holds.  This lower bound is tight since RSDM for strict preference lists is non-bossy. We moreover interpret our problem in terms of the classical secretary problem and prove that our mechanism provides the best randomised strategy of the administrator who interviews the applicants.
\end{abstract}

\noindent
{\bf Keywords:} House allocation problem; Assignment problem; Pareto optimal matching; Randomised mechanisms; Truthfulness

\section{Introduction}\label{sec:intro}
We study the problem of allocating a set of indivisible objects among a set of agents.  Each agent has private ordinal preferences over a subset of objects --- those they find acceptable, and each agent may be allocated at most one object. This problem has been studied by both economists and computer scientists. When monetary transfers are not permitted, the problem is referred to as the \emph{House Allocation problem} (henceforth abbreviated by HA \cite{HZ79,AS98,Zho90} or the \emph{Assignment problem} \cite{Gar73,Bogomolnaia-Moulin} in the literature. In this paper we opt for the term House Allocation problem.
Most of the work in the literature assumes that the agents have strict preferences over their acceptable objects. However, it often happens though that an agent is indifferent between two or more objects.  Here we let agents express indifference, and hence preferences may involve ties unless explicitly stated otherwise.

It is often desired that as many objects as possible become allocated among the agents --- i.e., an allocation of maximum size is picked, hence making as many agents happy as possible. Usually, depending on the application of the problem, we are required to consider some other optimality criteria, sometimes instead of, and sometimes in addition to, maximising the size of the allocation.
Several optimality criteria have been considered in the HA setting, and perhaps the most studied such concept is
\emph{Pareto optimality} (see, e.g., \cite{AS98,ACMM04,CG10a,CEFMMP14,SS15}), sometimes referred to as Pareto efficiency.  Economists, in particular, regard Pareto optimality as the most fundamental requirement for any ``reasonable'' solution to a non-cooperative game.  Roughly speaking, an allocation $\matching$ is Pareto optimal if there does not exist another allocation $\matching'$ in which no agent is worse off, and at least one agent is better off, in $\matching'$.  In this work we are mainly concerned with Pareto optimal allocations of maximum size, but will also consider weighted generalisations.

The related \emph{Housing Market} problem (HM) \cite{Rot82,RP77,SS74} is the variant of HA in which there is an \emph{initial endowment}, i.e., each agent owns a unique object initially (in this case the numbers of agents and objects are usually defined to be equal).  In this setting, the most widely studied solution concept is that of the \emph{core}, which is an allocation of agents to objects satisfying the property that no coalition $C$ of agents can improve (i.e., every agent in $C$ is better off) by exchanging their own resources (i.e., the objects they brought to the market).  In the case of strict preferences, the core is always non-empty \cite{RP77}, unique, and indeed Pareto optimal.  When preferences may include ties, the notion of \emph{core} that we defined is sometimes referred to as the \emph{weak core}. 
In this case a core allocation need not be Pareto optimal. Jaramillo and Manjunath~\cite{JM12}, Plaxton~\cite{P-2013}, and Saban and Sethuraman~\cite{Saban-Sethuraman} provide polynomial-time algorithms
for finding a core allocation that does additionally satisfy Pareto optimality. Our problem differs from HM in that there is no initial endowment, and hence our focus is on Pareto optimal matchings rather than outcomes in the core.

For strictly-ordered preference lists, Abraham et al.\ \cite{ACMM04} gave a characterisation of Pareto optimal matchings that led to an $O(m)$ algorithm for checking an arbitrary matching for Pareto optimality, where $m$ is the total length of the agents' preference lists.  This characterisation was extended to the case that preference lists may include ties by Manlove \cite[Sec.\ 6.2.2.1]{Man13}, also leading to an $O(m)$ algorithm for checking a matching for Pareto optimality.  For strictly-ordered lists, a maximum cardinality Pareto optimal matching can be found in $O(\sqrt{n_1}m)$ time, where $n_1$ is the number of agents \cite{ACMM04}.  The fastest algorithm currently known for this problem when preference lists may include ties is based on minimum cost maximum cardinality matching and has complexity $O(\sqrt{n}m\log n)$ (see, e.g., \cite[Sec.\ 6.2.2.1]{Man13}, where $n$ is the total number of agents and objects.

As stated earlier, agents' preferences are private knowledge. Hence, unless they find it in their own best interests, agents may not reveal their preferences truthfully. An allocation mechanism is \emph{truthful} if it gives agents no incentive to misrepresent their preferences. Perhaps unsurprisingly,
a mechanism based on constructing a maximum cardinality Pareto optimal allocation is manipulable by agents misrepresenting their preferences (Theorem~\ref{thm:min-cost-rank-maximal} in Section~\ref{sec:setting}). Hence, we need to make a compromise and weaken at least one of these requirements. In this work, we relax our quest for finding a maximum cardinality Pareto optimal allocation by trading off the size of a Pareto optimal allocation against truthfulness; more specifically, we seek truthful Pareto optimal mechanisms that provide good approximations to the maximum size.

Under strict preferences, Pareto optimal matchings can be computed by a classical algorithm called the \emph{Serial Dictatorship Mechanism} (SDM) (see, e.g., \cite{AS98}), also referred to as the \emph{Priority Mechanism} (see, e.g., \cite{Bogomolnaia-Moulin}). SDM is a straightforward greedy algorithm that takes each agent in turn and allocates to him the most preferred available object on his
preference list. Precisely due to this greedy approach, SDM is truthful. Furthermore, SDM
is guaranteed to find a Pareto optimal allocation that has size at least half that of a maximum one, merely because \emph{any} Pareto optimal allocation has size at least half that of a maximum one (see, e.g., \cite{ACMM04}). Hence, at least in the case of strict preferences, we are guaranteed an approximation ratio of $2$. Can we do better? It turns out that if we stay in the realm of deterministic mechanisms, a 2-approximation is the best we can hope for (Theorem ~\ref{thm:min-cost-rank-maximal}, Section~\ref{sec:setting}).


Hence we turn to randomised mechanisms in order to achieve a better approximation ratio. The obvious candidate to consider is the Random Serial Dictatorship Mechanism (RSDM) (see, e.g., \cite{AS98}), also known as the Random Priority mechanism (see, e.g., \cite{Bogomolnaia-Moulin}), that is defined for HA instances with strict preferences.  RSDM randomly generates an ordering of the agents and then proceeds by running SDM relative to this ordering. 

When indifference is allowed, finding a Pareto optimal allocation is not as straightforward as for strict preferences. For example, one may consider breaking the ties randomly and then applying SDM. This approach, unfortunately, may produce an allocation that is not Pareto optimal.
To see this, consider a setting with two agents, $\agentone$ and $\agenttwo$, and two objects, $\object_1$ and $\object_2$.  Assume that $\agentone$ finds both objects acceptable and is indifferent between them, and that $\agenttwo$ finds only $\object_1$ acceptable. The only Pareto optimal matching for this setting is the one in which $\agentone$ and $\agenttwo$ are assigned $\object_2$ and $\object_1$ respectively. Assume that $\agentone$ is served first and that, as both objects are equally acceptable to him, is assigned $\object_1$ (after an arbitrary tie-breaking). Therefore when $\agenttwo$'s turn arrives, there is no object remaining that he finds acceptable, and is hence left unmatched, resulting in a matching that is not Pareto optimal.

Few works in the literature have considered extensions of SDM to the case where agents' preferences may include ties. 
However Bogomolnaia and Moulin \cite{BM-2004} and Svensson \cite{Svensson} provide an implicit extension of SDM (in the former case for \emph{dichotomous preferences}\footnote{An agent's preference list is \emph{dichotomous} if it comprises a single tie containing all acceptable objects.}) but do not give an explicit description of an algorithm. Aziz et al.\ \cite{ABH-2013} provide an explicit extension for a more general class of problems, including HA.
Pareto optimal matchings in HA can also be found by reducing to the HM setting~\cite{JM12}, which involves creating dummy objects as endowments for the agents.  This allows one of the aforementioned algorithms for HM \cite{JM12,P-2013,Saban-Sethuraman} to be utilised to find a Pareto optimal matching in the core.  However the reduction increases the instance size, and in particular the number of agents $\numagents$ and the maximum length of a tie in any agent's preference list.  Consequently, even the fastest truthful Pareto optimal mechanism for HM, that of Saban and Sethuraman~\cite{Saban-Sethuraman}, has time complexity no better than $O(n_1^3)$ in the worst case.

\paragraph{Contributions of this paper.}
In this paper we provide a natural and explicit extension of SDM for the setting in which preferences may exhibit indifference.  We argue that our extension is more intuitive than that of Aziz et al.\ \cite{ABH-2013} when considering specifically HA.  Moreover, as the mechanism of Saban and Sethuraman \cite{Saban-Sethuraman} does not consider the agents sequentially, it is difficult to analyse its approximation guarantee.  Our algorithm runs in time $O(n_1^2\gamma\david{+m})$, where $\gamma$ is the maximum length of a tie in any agent's preference list \david{and $m$ is the total length of the agents' preference lists}.  
This is faster than the algorithm in \cite{Saban-Sethuraman} when $\gamma = o(n_1)$ \david{and $m=o(n_1^3)$}.
We prove the following results that involve upper and lower bounds for the approximation ratio (relative to the size of a maximum cardinality Pareto optimal matching) of randomised truthful mechanisms for computing a Pareto optimal matching:

\begin{enumerate}
	\item By extending RSDM to the case of preference lists with ties, we give a universally truthful randomised mechanism\footnote{A randomised mechanism is universally truthful if it is a probability distribution over truthful deterministic mechanisms. This is the strongest known notion of truthfulness for randomised mechanisms.} for finding a Pareto optimal matching that has an approximation ratio of $\frac{e}{e-1}$ with respect to the size of a maximum cardinality Pareto optimal matching.
	\item We give a lower bound of $\frac{18}{13}$ on the approximation ratio of any universally truthful Pareto optimal mechanism in settings with strict preferences. If the mechanism must additionally be non-bossy, \piotr{with an additional technical assumption}\footnote{A deterministic mechanism in settings with strict preferences is non-bossy if no agent can misreport his preferences in such a way that his allocation is not changed but the allocation of some other agent is changed. \piotr{This additional assumption restricts the class to randomised mechanisms that are symmetrizations of truthful, Pareto optimal and non-bossy mechanisms, see Section 6 for details.}}, we observe that \cite{Ba18} implies an improved lower bound of $\frac{e}{e-1}$.  This lower bound is tight since the classical RSDM mechanism for strict preferences is non-bossy.
	\item We extend RSDM to the setting where agents have priorities (weights) and our goal is to find a maximum weight (as opposed to maximum cardinality) Pareto optimal matching. Our mechanism is universally truthful and guarantees a $\frac{e}{e-1}$-approximation with respect to the weight of a maximum weight Pareto optimal matching.
  \item We finally observe that our problem has an ``online'' or sequential flavour similar to secretary problems\footnote{In the secretary problem, an administrator is willing to hire the best secretary out of $n$ rankable applicants for a position. The applicants are interviewed one-by-one in random order. A decision about each particular applicant is to be made immediately after the interview. Once rejected, an applicant cannot be recalled. During the interview, the administrator can rank the applicant among all applicants interviewed so far, but is unaware of the quality of yet unseen applicants. The question is about the optimal strategy to maximise the probability of selecting the best applicant.}. Given this interpretation, we prove that our mechanism uses the best random strategy of interviewing the applicants in the sense that any other strategy would lead to an approximation ratio worse than $\frac{e}{e-1}$ (see also below under related work).
\end{enumerate}

\paragraph{Discussion of technical contributions.}
As observed above via a simple example, SDM with arbitrary tie breaking need not lead to a Pareto optimal matching in general.  Indeed, the presence of indifference in agents' preference lists introduces major technical difficulties.  This is because decisions with respect to objects in one tie cannot be committed to when an agent is considered, as they may block some choices for future agents. When extending SDM from strict preferences to preferences with ties, we first present an intuitive mechanism, called SDMT-1, based on the idea of augmenting paths. It is relatively easy to prove that SDMT-1 is Pareto optimal and truthful. We also show that SDMT-1 is able to generate any given Pareto optimal matching. However, it is difficult to analyse the approximation guarantee of the randomised version of SDMT-1. For this purpose we build on the primal-dual analysis of Devanur et al.\ \cite{Devanur-etal}. They employ a linear programming (LP) relaxation of the bipartite weighted matching problem. They prove that their dual solution is feasible in expectation for the dual LP and use it to show the approximation guarantee. Towards this goal they prove two technical lemmas, a dominance lemma and a monotonicity lemma. The randomised version of SDMT-1 uses random variables $Y_i$ for each agent $i \in N$ to generate a random order in which agents are considered. Considering agent $i$ and fixed values of the random variables $Y_{-i}$ of all other agents, Devanur et al.\ \cite{Devanur-etal} define a threshold $y^c$, which as $Y_i$ varies determines when agent $i$ is matched (to an object) or unmatched (dominance lemma). (Note that we will denote the threshold $y^c$ as $\theta$.) The monotonicity lemma shows how values of the dual LP variables change when $Y_i$ varies. To extend the definition of $y^c$, we need to remember the structure of all potential augmenting paths in SDMT-1, and for this purpose we introduce a second mechanism,
SDMT-2.
Interestingly, SDMT-2 is inspired by the idea of top trading cycle mechanisms, see, e.g., \cite{Saban-Sethuraman}, however it retains the ``sequential'' nature of SDMT-1. The two mechanisms, SDMT-1 and SDMT-2, are equivalent: they match the same agents, giving them objects from the same ties. This implies that SDMT-2 is also truthful and Pareto optimal. SDMT-2 is the key to defining the threshold $y^c$: its running time is no worse than that of SDMT-1, but it implicitly maintains all relevant augmenting paths arising from agents' ties. We prove the monotonicity and dominance lemmas for SDMT-2 by carefully analysing the structure of frozen subgraphs that are generated from the relevant augmenting paths; here frozen roughly means that they will not change subsequently. 
Finally, we would like to highlight that our proof of an $\frac{18}{13}$ lower bound on the approximation ratio of any universally truthful and Pareto optimal mechanism uses Yao's minmax principle and an interesting case analysis to account for all such possible mechanisms. 

\begin{paragraph}{Related work.}
This work can be placed in the context of designing truthful approximate mechanisms for problems in the absence of monetary transfer \cite{PT13}.  \hide{We now survey some of the prior work in this area that is most relevant to our problem.}Bogomolnaia and Moulin \cite{Bogomolnaia-Moulin} designed a randomised weakly truthful and envy-free mechanism, called the Probabilistic Serial mechanism (PS), for HA with complete lists.  Very recently the same authors considered the same approximation problem as ours but in the context of envy-free rather than truthful mechanisms, and for strict preference lists and unweighted agents~\cite{BM15}.  They showed that PS has an approximation ratio of $\frac{e}{e-1}$, which is tight for any envy-free mechanism. Bhalgat et al.\ \cite{Bhalgat-etal} investigated the social welfare of PS and RSDM. Tight deterministic truthful mechanisms for weighted matching markets were proposed by Dughmi and Ghosh \cite{Dughmi-Ghosh} and they also presented an $O(\log n)$-approximate random truthful mechanism for the Generalised Assignment Problem (GAP) by reducing, with logarithmic loss in the approximation, to the solution for the value-invariant GAP. In subsequent work Che et al.\ \cite{CGL-2014} provided an $O(1)$-approximation mechanism for GAP. Aziz et al.\ \cite{AGMW15} studied notions of fairness involving the stochastic dominance relation in the context of HA, and presented various complexity results for problems involving checking whether a fair assignment exists.
Chakrabarty and Swamy \cite{CW14} proposed \emph{rank approximation} as a measure of the quality of an outcome and introduced the concept of \emph{lex-truthfulness} as a notion of truthfulness for randomised mechanisms in HA.

RSDM is related to online bipartite matching algorithms. The connection was observed by Bhalgat et al.\ \cite{Bhalgat-etal}, who noted the similarity between RSDM and the RANKING algorithm of \david{Karp et al.\ \cite{kvv}. Karp et al.\ \cite{kvv} proved that the expected size of the matching given by their RANKING algorithm is at least $\frac{e-1}{e}$ times the optimal size.  Bhalgat et al.\ \cite{Bhalgat-etal} observed that RSDM will essentially behave the same way as RANKING for instances of HA where the agents' preference lists relate to the order in which the objects arrive. Hence for this family of instances an approximation ratio of $\frac{e}{e-1}$ holds for RSDM.}

The weighted version of our problem is related to two widely-studied online settings, known in the literature as the online vertex-weighted bipartite matching problem \cite{Aggarwal-etal} and secretary problems \cite{BIKK08}.  In our problem the administrator holds all the objects (they can be thought of as available positions), and all agents with unknown preference lists are applicants for these objects. Each applicant also has a private weight, which can be thought of as their quality (reflecting the fact that some of an agent's skills may not be evident from their CV, for example). However we assume that they cannot overstate their weights (skills), because they might be checked and punished. This is similar to the classical assumption of no overbidding (e.g., in sponsored search auctions).  Applicants are interviewed one-by-one in a random order. When an applicant arrives he chooses his most-preferred available object and the decision as to whether it is allocated to him is made immediately, and cannot be changed in the future. \hide{ The administrator is unaware as to whether applicants still to arrive will be matched or not.} \hide{
\piotr{We would like to point out here that the aspect of private weights of the agents with the no-overbidding assumption is not the main emphasis and contribution of our paper.}}

Our weighted agents correspond to weighted vertices in the vertex-weighted bipartite matching context, but our objects do not arrive online as in the setting of \cite{Aggarwal-etal}. 
However, if the preference ordering of each agent in our setting, over his acceptable objects,  coincides with the arrival order of the objects in \cite{Aggarwal-etal}, then the two problems are the same.
In the transversal matroid secretary problem, see, e.g., \cite{DimitrovP12}, objects are known in advance as in our setting, weighted agents arrive in a (uniform) random order, and the goal is to match them to previously unmatched objects. 
The administrator's goal is to find a (random) arrival order of agents that maximises the ratio between the total weight of matched agents and the maximum weight of a matching if all the agents’ preference lists are known in advance.
We show that even if the weights of all agents are the same, our algorithm uses the best possible random strategy; no other such strategy leads to better than $\frac{e}{e-1}$-approximate matching.
\end{paragraph}

\begin{paragraph}{Organisation of the paper.}
The remainder of the paper is organised as follows. In Section 2 we define notation and terminology used in this paper, and show the straightforward lower bound for the approximation ratio of deterministic truthful mechanisms.  SDMT-1 and SDMT-2 are presented in Sections 3 and 4 respectively, and in the latter section it is proved that the two mechanisms are essentially equivalent.  The approximation ratio of $\frac{e}{e-1}$ for the randomised version of the two mechanisms is established in Section 5, whilst Section 6 contains our lower bound results. Finally, some concluding remarks are given in Section 7.
\end{paragraph}

\section{Definitions and preliminary observations} \label{sec:setting}
Let $\agentset=\{1,2,\cdots,\numagents\}$ be a set of $\numagents$ \emph{agents} and $\objectset=\david{\{\object_1,\object_2,\dots,\object_{\numobjects}\}}$ be a set of $\numobjects$ \emph{objects}. Let $\numvertices=\numagents+\numobjects$. Let $[i]$ denote the set $\{1,2,\cdots,i\}$.
We assume that each agent $\agent\in \agentset$ finds a subset of objects acceptable and has a preference ordering, not necessarily strict, over these objects.
We write $\object_t \strictlypref_{\agent} \object_s$ to denote that agent $\agent$ \emph{strictly prefers} object $\object_t$ to object $\object_s$, and write $\object_t \tie_{\agent} \object_s$ to denote that $\agent$ is \emph{indifferent between} $\object_t$ and $\object_s$. We use $\object_t \weaklypref_{\agent} \object_s$ to denote that agent $\agent$ either strictly prefers $\object_t$ to $\object_s$ or is indifferent between them, and say that $\agent$ \emph{weakly prefers} $\object_t$ to $\object_s$.
In some cases a weight $\weight_{\agent}$ is associated with each agent $\agent$, representing the priority or importance of the agent.  Weights need not be distinct. Let $\weightprofile = (\weight_1, \weight_2, \ldots, \weight_{\numagents})$.
To simplify definitions, we assume that all agents are assigned weight equal to $1$ if we are in an unweighted setting.

We assume that the indifference relation is transitive. This implies that each agent essentially divides his acceptable objects into different bins or \emph{indifference classes} such that he is indifferent between the objects in the same indifference class and has a strict preference ordering over these indifference classes. For each agent $\agent$, let $\eqclass^{\agent}_k$, $1 \leq k \leq \numobjects$, denote the $k$th indifference class, or \emph{tie}, of agent $\agent$. We also assume that if there exists $l\in[\numobjects]$, where $\eqclass^{\agent}_l = \emptyset$, then $\eqclass^{\agent}_q = \emptyset$, $\forall q, l \leq q\leq \numobjects$. We let $\preflist(\agent) = (\eqclass^{\agent}_1 \strictlypref_{\agent} \eqclass^{\agent}_2 \strictlypref_{\agent}\cdots \strictlypref_{\agent} \eqclass^{\agent}_{\numobjects})$ and call $\preflist(\agent)$ the \emph{preference list} of agent $\agent$.
We abuse notation and write $\object\in \preflist(\agent)$ if $\object$ appears in preference list $\preflist(\agent)$, i.e., if agent $\agent$ finds object $\object$ acceptable.
We say that agent $\agent$ ranks object $\object$ in $k$th position if
$\object\in \eqclass^{\agent}_k$.
We denote by $\rank(\agent,\object)$ the rank of object $\object$ in agent $\agent$'s preference list and let $\rank(\agent,\object)=\numobjects+1$ if $\object$ is not acceptable to $\agent$.
%
Therefore $\object_t \strictlypref_{\agent} \object_s$ if and only if $\rank(\agent,\object_t) < \rank(\agent,\object_s)$, and $\object_t \tie_{\agent} \object_s$ if and only if $\rank(\agent,\object_t) = \rank(\agent,\object_s)$.

Let $\preflist=(\preflist(1),\preflist(2),\cdots,\preflist(\numagents))$ denote the joint preference list profile of all agents. We write $\preflist(-\agent)$ to denote the joint preference list profile of all agents except agent $\agent$; i.e., $\preflist(-\agent) = (\preflist(1),\ldots,\preflist(\agent-1),\preflist(\agent+1),\ldots,\preflist(\numagents))$. Let $\preflistset$ denote the set of all possible joint preference list profiles. An instance of HA is denoted by $\inputsetting=(\agentset, \objectset, \preflist, \weightprofile)$. We drop $\weightprofile$ and write $\inputsetting=(\agentset, \objectset, \preflist)$ if we are dealing with an instance where agents are not assigned weights, or equivalently if they all have the same weight. Let $\inputsettingset$
denote the set of all possible instances of HA.


A \emph{matching} $\matching$ is a subset of $N\times A$ such that each agent and object appears in at most one pair of $\matching$.  If $(\agent,\object) \in \matching$, agent $\agent$ and object $\object$ are said to be \emph{matched together}, and $\object$ is the \emph{partner} of $\agent$ and vice versa.
If $(\agent,\object)\in \matching$ for some $\object$, we say that $\agent$ is \emph{matched}, and \emph{unmatched} otherwise.  The definitions of \emph{matched} and \emph{unmatched} for an object are analogous.  If agent $\agent$ is matched, $\matching(\agent)$ denotes the object matched to $\agent$.  Similarly if object $\object$ is matched, $\matching^{-1}(\object)$ denotes the agent matched to $\object$. In what follows, we will refer to the \emph{underlying graph} of $\inputsetting$, which is the undirected graph $G_0=(V,E)$
where $V=\agentset\cup \objectset$ and  $E=\{(\agent,\object),\agent\in \agentset,\object\in \preflist(\agent)\}$.
We also use $\matching$ to denote a matching (in the standard graph-theoretic sense) in $G_0$.
The \emph{size} of a matching $\matching$ is equal to the number of agents matched under $\matching$. In the presence of weights, the \emph{weight} of a matching is equal to the sum of the weights of the matched agents.

For two given matchings $\matching_1, \matching_2$, we will frequently use $\matching_1 \oplus \matching_2$ to denote the symmetric difference with respect to their sets of edges.
%
%
An \emph{alternating path} in $G_0$, given a matching $\matching_1$, is a path that consists of edges that alternately belong to $\matching_1$ and do not belong to $\matching_1$. An \emph{augmenting path} in $G_0$ is an alternating path where the first and the last
vertices on the path are unmatched in $\matching_1$.
To \emph{augment} along an
augmenting path,
given matching $\matching_1$, means that a new matching $\matching_2$ is created by removing edges on the path that belong to $\matching_1$ and adding edges on the path that do not belong to $\matching_1$.

A matching $\matching$ is \emph{Pareto optimal} if there is no other matching under which some agent is better off while none is worse off. Formally, $\matching$ is \emph{Pareto optimal} if there is no other matching $\matching'$ such that (i) $\matching'(\agent) \weaklypref_{\agent} \matching(\agent)$ for all $\agent \in \agentset$, and (ii) $\matching'(\agent') \strictlypref_{\agent'} \matching(\agent')$ for some agent $\agent' \in \agentset$. Manlove \cite[Sec.\ 6.2.2.1]{Man13} gave a characterisation of Pareto optimal matchings in instances of HA (potentially with ties) in terms of a number of graph-theoretic structures, which we will now define.

An \emph{alternating path coalition} w.r.t.\ $\matching$ comprises a sequence $P=\left\langle i_0,i_1,\ldots, i_{r-1}, \object_k\right\rangle$, for some $r\geq1$, where $i_j$ is a matched agent ($0\leq j \leq r-1$) and $\object_k$ is an unmatched object. If $r=1$ then $i_0$ strictly prefers $\object_k$ to $\matching(i_0)$. Otherwise, if $r\geq 2$, $i_0$ strictly prefers $\matching(i_1)$ to $\matching(i_0)$, $i_j$ weakly prefers $\matching(i_{j+1})$ to $\matching(i_j)$ ($1\leq j \leq r-2$), and $i_{r-1}$ weakly prefers $\object_k$ to $\matching(i_{r-1})$.

An \emph{augmenting path coalition} w.r.t.\ $\matching$ comprises a sequence $P=\left\langle i_0,i_1,\ldots, i_{r-1}, \object_k\right\rangle$, for some $r\geq1$, where $i_0$ is an unmatched agent and $\object_k$ is an unmatched object. If $r=1$ then $i_0$ finds $\object_k$ acceptable. Otherwise, if $r\geq 2$, $i_j$ is a matched agent ($1\leq j \leq r-1$), $i_0$ finds $\matching(i_1)$ acceptable, $i_j$ weakly prefers $\matching(i_{j+1})$ to $\matching(i_j)$ ($1\leq j \leq r-2$), and $i_{r-1}$ weakly prefers $\object_k$ to $\matching(i_{r-1})$.

A \emph{cyclic coalition} w.r.t.\ $\matching$ comprises a sequence of applicants $P=\left\langle i_0,i_1,\ldots, i_{r-1} \right\rangle$, for some $r\geq 2$, all matched in $\matching$, such that $i_j$ weakly prefers $\matching(i_{j+1})$ to $\matching(i_j)$ for each $j$ ($0\leq j \leq r-1$), and $i_j$ strictly prefers $\matching(i_{j+1})$ to $\matching(i_j)$ for some $j$ ($0\leq j \leq r-1$) (all subscripts are taken modulo $r$ when reasoning about cyclic coalitions).

\begin{proposition}[\cite{Man13}]
Given an instance $\inputsetting$ of HA and a matching $\matching$ in $\inputsetting$, $\matching$ is Pareto optimal if and only if $\matching$ admits no \emph{alternating path coalition}, no \emph{augmenting path coalition}, and no \emph{cyclic coalition}.
\label{prop:characterisation}
\end{proposition}
%
%

Let $\matchingset$ denote the set of all possible matchings. A \emph{deterministic mechanism} $\mechanism$ maps an instance of HA to a matching, i.e., $\mechanism : \inputsettingset \rightarrow \matchingset$. Let
$\randomlist: \matchingset \rightarrow [0,1]$ denote a distribution over possible matchings (which we also call a \emph{random matching}); i.e., $\sum_{\matching\in\matchingset}\randomlist(\matching) = 1$.
A \emph{randomised mechanism} $\mechanism$ is a mapping from $\inputsettingset$ to a distribution over possible matchings, i.e., $\mechanism:\inputsettingset\rightarrow Rand(\matchingset)$, where $Rand(\matchingset)$ is the set of all random matchings.
A deterministic mechanism is \emph{Pareto optimal} if it always returns a Pareto optimal matching. A randomised mechanism is \emph{Pareto optimal} if it always returns a distribution over Pareto optimal matchings.

Agents' preferences are private knowledge and an agent may prefer not to reveal his preferences truthfully if it is not in his best interests, for a given mechanism.
A deterministic mechanism $\mechanism$ is \emph{dominant strategy truthful} (or \emph{truthful}) if agents always find it in their best interests to declare their preferences truthfully, no matter what other agents declare, i.e., for every joint preference list profile $L$, for every agent $\agent$, and for every possible declared preference list $\preflist'(\agent)$ for $\agent$, $\mechanism(\preflist(\agent),\preflist(-\agent))\weaklypref_{\agent}\mechanism(\preflist'(\agent),\preflist(-\agent))$. 
%
A randomised mechanism $\mechanism$ is \emph{universally truthful}
if it is a probability distribution over deterministic truthful mechanisms.
%

Denote by $\weight(\mechanism(\inputsetting))$ the (expected) weight of the (random) matching generated by mechanism $\mechanism$ on instance $\inputsetting \in \inputsettingset$, and by $\weight(\inputsetting)$ the weight of a maximum weight 
matching in $\inputsetting$.
The \emph{approximation ratio} of $\mechanism$ is then  defined as
$\max_{\inputsetting\in\inputsettingset}\frac{\weight(\inputsetting)}{\weight(\mechanism(\inputsetting))}$. 
Note that a maximum weight matching has the same weight as a maximum weight Pareto optimal matching\david{, as the following proposition shows.}
\begin{proposition}
\david{Given an instance $\inputsetting$ of HA, a maximum weight matching has the same weight as a maximum weight Pareto optimal matching.}
\label{prop:weight}
\end{proposition}
\begin{proof}
\david{We provide a procedure for transforming a maximum weight matching $\mu$ in $\inputsetting$ to a Pareto optimal matching $\matching'$ in $\inputsetting$ with the same weight.}

\david{Let $G_0$ be the underlying graph for $\inputsetting$ and let $G'_0$ be the subgraph of $G_0$ induced by $N'\cup A$, where $N'$ is the set of agents who are matched in $\mu$.  Define the cost of each edge $(\agent,\object_j)$ in $G_0'$ to be $rank(\agent,\object_j)$. Find a maximum cardinality matching $\mu'$ of minimum cost in $G_0'$.  It is easy to see that $\mu$ and $\mu'$ are of the same cardinality and have the same weight, as they each match all agents in $N'$. It remains to show that $\mu'$ is Pareto optimal in $\inputsetting$.}

\david{If $\matching'$ is not Pareto optimal in $\inputsetting$ then by Proposition \ref{prop:characterisation}, $\matching'$ admits a coalition $C$ that is either an alternating path coalition, or an augmenting path coalition, or a cyclic coalition.  If $C$ is an augmenting path coalition then $\matching'\oplus C$ has larger weight than $\matching$, a contradiction as $\matching$ is a maximum weight matching.  Hence $C$ is an alternating path coalition or a cyclic coalition.  In either case let $\matching''=\matching'\oplus C$.  Then $|\matching''|=|\matching'|$ but the cost of $\matching''$ is less than the cost of $\matching'$, a contradiction as $\matching'$ is a maximum cardinality minimum cost matching in $G_0'$.  Hence $\matching'$ is Pareto optimal in $\inputsetting$.}
\end{proof}

We now give a straightforward lower bound for the approximation ratio of any deterministic truthful mechanism for HA with strict preferences.
\begin{theorem}\label{thm:min-cost-rank-maximal}
No deterministic truthful 
mechanism for HA can achieve approximation ratio better than $2$.  The result holds even for strict preferences.
\end{theorem}
\begin{proof}
Consider \david{an HA instance $\inputsetting$} with two agents, $\agentone$ and $\agenttwo$, and two objects, $\object_1$ and $\object_2$. Assume that both agents have weight $1$ and \david{strictly} prefer $\object_1$ to $\object_2$.  Then \david{$\inputsetting$} admits two 
matchings of size (weight) 2.
Assume, for a contradiction, that there exists a truthful 
mechanism $\mechanism$ with approximation ratio strictly smaller than $2$. \david{Then in $\inputsetting$,} $\mechanism$ must pick one of the two matchings of size 2. Assume, without loss of generality, that $\mechanism$ picks $\matching =\{(\agentone,\object_2),(\agenttwo,\object_1)\}$.
Now, assume that \david{agent} $\agentone$ misrepresents his acceptable objects and declares $\object_1$ as the only object acceptable to him.  \david{Let $\inputsetting'$ denote the instance of HA so obtained}.  As $\mechanism$ is truthful, \david{when executed on $\inputsetting'$} it must not assign $\object_1$ to $\agentone$, or else $\agentone$ finds it in his best interests to misrepresent his preferences as he \david{would strictly} prefer his allocated object \david{in $\inputsetting'$} to his allocated object \david{in $\inputsetting$}. Hence $\mechanism$ must return a matching of size at most 1 (by assigning an object to \david{agent} $\agenttwo$) \david{when applied to} $\inputsetting'$. However, \david{$\inputsetting'$} admits a 
matching of size 2, \david{namely} $\matching' = \{(\agentone,\object_1),(\agenttwo,\object_2)\}$. Therefore the approximation ratio of $\mechanism$ is \david{at least} 2, a contradiction.
\end{proof}
\begin{corollary}
No deterministic truthful Pareto optimal mechanism for HA can achieve approximation ratio better than $2$.  The result holds even for strict preferences.
\end{corollary}

As mentioned in Section 1, the upper bound of 2 is achievable via SDM for HA with strict preferences \david{\cite{ACMM04}}.
If weights and ties exist, simply ordering the agents in decreasing order of their weights and running SDMT-1 (see Algorithm \ref{alg:SDMT-1-1} in Section 3) or SDMT-2 (see Algorithm \ref{alg:indifference} in Section 4) gives a deterministic truthful and Pareto optimal mechanism with approximation ratio $2$ (Theorem~\ref{thm:sdmt-apx-ratio} in Section~\ref{sec:sdt-aug}).
This resolves the problem for deterministic mechanisms and motivates looking into relaxing our requirements. In the following sections we look for randomised truthful mechanisms that construct  `large' weight Pareto optimal matchings.



\section{First truthful mechanism: SDMT-1}\label{sec:sdt-aug}
\subsection{Introduction}
When preferences are strict, SDM produces a Pareto optimal matching. 
However when indifference is allowed, finding an arbitrary Pareto optimal matching is not as straightforward as in the case of strict preferences, as illustrated via an example in Section \ref{sec:intro}.

In Section \ref{sec:sdmt1} we introduce SDMT-1, Serial Dictatorship Mechanism with Ties, a mechanism that generalises SDM to the case where agents' preferences may involve ties. Then in Section \ref{sec:sdmt1prop}, we show that SDMT-1 is truthful and is guaranteed to produce a Pareto optimal matching.  We further show that SDMT-1 is capable of generating any given Pareto optimal matching.

\subsection{Mechanism SDMT-1}
\label{sec:sdmt1}
Let $\inputsetting=(\agentset,\objectset,\preflist)$ be an instance of HA, and let 
a fixed order $\agentsorder$ of the agents be given. Assume, w.l.o.g., that $\agentsorder(\agent) = \agent$ for all agents $\agent \in \agentset$.
The formal description of SDMT-1 is given in Algorithm~\ref{alg:SDMT-1-1}; an informal description follows.

SDMT-1 constructs an undirected bipartite graph $G=(V,E)$ where  $V=\agentset\cup\machineset$ and the set of edges $E$ changes during the execution of SDMT-1; initially $E = \emptyset$.  The mechanism returns a matching $\matching$; initially $\matching=\emptyset$.  It then proceeds in $\numagents$ phases, where each phase corresponds to one iteration of the for loop in Algorithm~\ref{alg:SDMT-1-1}.  In phase $\agent$, agent $\agent$ is considered and the objects in $\agent$'s preference list are examined in the order of the indifference classes they belong to.  Recall that $\eqclass^{\agent}_{\ell}$ denotes the $\ell$'th indifference class of agent $\agent$.  When objects $\object \in \eqclass^{\agent}_{\ell}$ are examined, edges $(\agent,\object)$ are provisionally added to $E$ for all $\object \in \eqclass^{\agent}_{\ell}$.  We then check whether \david{$\matching$ admits} an augmenting path in $G$ that starts from agent $\agent$. If such a path exists, we augment along that path and modify $\matching$ accordingly. This would mean that agent $\agent$ is assigned some $\object\in \eqclass^{\agent}_{\ell}$ and every other agent already matched is assigned an object that he ranks in the same indifference class his previous object. Otherwise -- if \david{$\matching$} admits no augmenting path in \david{$G$} -- edges $(\agent,\object)$ are removed from $E$ for all $\object \in \eqclass^{\agent}_{\ell}$.  In general, once an agent $\agent$ is assigned an object $\object\in \eqclass^{\agent}_{\ell}$ he will remain matched in $\mu$, although he may be required to exchange $\object$ for another object in $\eqclass^{\agent}_{\ell}$ in order to allow a newly-arrived agent to receive $\object$. 

\begin{algorithm}[t!]
\SetAlgoLined
\caption {Serial Dictatorship Mechanism with Ties,  version 1 (SDMT-1)}
\KwIn{ Agents $\agentset$; Objects $\objectset$; Preference list profile $\preflist$; An order of agents $\agentsorder$}
\KwOut{Matching $\matching$}
Let $G= (\agentset \cup \objectset, E)$, $E \leftarrow \emptyset$, $\matching \leftarrow \emptyset$. \\
\For {each agent $i\in \agentset$ in the order of $\agentsorder$}{
 Let $\ell \leftarrow 1$ \\
 Step (*): \If {$\eqclass^i_{\ell} \not = \emptyset$} 
{$E \leftarrow E \cup \{(i, \object): \object \in \eqclass^i_{\ell}\}$; // all new edges are non-matching edges\\
 \If {{\normalfont there is an augmenting path from} $i$ {\normalfont in} $G$} {augment along this path and update $\matching$ accordingly; // $i$ is provisionally allocated some $\object \in \eqclass^i_{\ell}$ and $(\agent,\object)$ is now a matching edge\\
}

 \Else {$E \leftarrow E \setminus \{(i, \object): \object \in \eqclass^i_{\ell}\}$ \\
        $\ell \leftarrow \ell + 1$; Go to Step (*)}
}
}
Return $\matching$; //each matched agent is allocated his matched object \\
\label{alg:SDMT-1-1}
\end{algorithm}

Notice that, at any stage of the mechanism, an edge $(\agent,\object)$ belongs to $E$ if and only if \emph{either} agent $\agent$ is matched in $\matching$ and $\object \tie_{\agent} \matching(\agent)$, \emph{or} SDMT-1 is at phase $\agent$ and examining the indifference class to which $\object$ belongs.
Therefore, it is fairly straightforward to observe the following.
\begin{observation}\label{obs:1}
At the end of phase $\agent$ of SDMT-1,
if agent $\agent$ is assigned no object then he will be assigned no object when SDMT-1 terminates. Otherwise, if $\agent$ is provisionally assigned an object $\object$, 
then he will be allocated an object that he ranks the same as $\object$ 
in the final matching.
\end{observation}

\subsection{Properties of SDMT-1}
\label{sec:sdmt1prop}
Before proceeding to prove our main claim, namely that SDMT-1 is truthful and produces a Pareto optimal matching, let us discuss a relevant concept that is both interesting in its own right and useful in the proofs that follow.  In practice agents may have priorities and the mechanism designer may wish to ensure that the agents with higher priorities are served before satisfying those with lower priorities. Roth et al.\ \cite{RSU05} studied this concept under the term \emph{priority matchings} in the case where each agent's preference list is one \david{single} tie. This work was motivated by the kidney exchange problem in which patients are assigned priorities based on various criteria; e.g., children and hard-to-match patients have higher priorities. Prior to Roth et al.\ \cite{RSU05}, Svensson \cite{Svensson} studied a similar concept under the name \emph{queue allocation} in a setting similar to ours.
We formally define this concept using the terminology \emph{strong priority matching}, reflecting both the definition in \cite{RSU05} and the fact that preference lists are more general than single ties.



\david{In general, assume that we are given} an ordering of the agents $\agentsorder=\agent_{1},\ldots,\agent_{\numagents}$.  \david{However, recall that in this section we are assuming, without loss of generality, that $\agent_j=j$, i.e., $\agentsorder=1,2,\dots,n_1$.}
For each matching $\matching$, the \emph{signature} of $\matching$ \david{w.r.t.\ $\agentsorder$}, denoted by \david{$\sig(\matching,\agentsorder)$},
is a vector $\left\langle \sig_{1},\ldots,\sig_{\numagents}\right\rangle$ where
%
for each $\david{\agent} \in [\numagents]$, \david{$\sig_\agent = \rank(\agent,\matching(\agent))$} if \david{$\agent$} is matched under $\matching$, and \david{$\sig_\agent = \numobjects +1$} otherwise.
A matching $\matching$ is a \emph{strong priority matching (SPM)} w.r.t.\ $\agentsorder$ if \david{$\sig(\matching,\agentsorder)$ is lexicographically minimum, taken over all matchings $\matching$}. \david{That is, (i) the highest priority agent $\agentone$ has one of his first-choice objects (assuming $L(\agentone)\neq \emptyset$); (ii) subject to (i), there is no matching $\matching'$ such that $\matching'(\agenttwo)\strictlypref_\agenttwo \matching(\agenttwo)$, where $\agenttwo$ is the agent with the second-highest priority; (iii) subject to (i) and (ii), there is no matching $\matching''$ such that $\matching''(\agentthree)\strictlypref_\agentthree \matching(\agentthree)$, where $\agentthree$ is the agent with the third-highest priority, etc.}  It is easy to see that a given \david{HA instance} may admit more than one SPM w.r.t.\ $\agentsorder$, but all of them have the same signature. When $\agentsorder$ is fixed and known, we simply say that $\matching$ is an SPM. 

\begin{theorem}\label{thm:ptm}
The matching produced by SDMT-1 is a strong priority matching w.r.t.\ $\agentsorder$.
\end{theorem}
\begin{proof}
\david{Let $\matching_k$ denote the matching at the end of phase $k$ (hence $\matching_{\numagents} = \matching$).} 
Assume, for a contradiction, that the claim does not hold. Hence $\matching$ is not an SPM in $\inputsetting$. Let $\matching^*$ be an SPM in $\inputsetting$. Let \david{$\agent$ be the first agent in $\agentsorder$ (i.e., the lowest-indexed agent)} who strictly prefers his partner under $\matching^*$ to his partner under $\matching$, i.e., \david{$\matching^*(\agent)\strictlypref_\agent\matching(\agent)$ and $\matching^*(j) \tie_{j} \matching(j)$, $\forall j < \agent$}
(we denote this fact by \textbf{D1}). Therefore, in phase \david{$\agent$} of SDMT-1 no augmenting path has been found starting from $(\agent,\object)$, for any object $\object$ such that $\object \weaklypref_\agent \matching^*(\agent)$ (we denote this fact by \textbf{D2}). 
Also, it follows from \textbf{D1} and Observation~\ref{obs:1} that, \david{$\matching^*(j) \tie_j \matching_{\agent-1}(j)$, $\forall j < \agent$} (we denote this fact by \textbf{D3}).

Let $G^*$ denote the graph $G$ in phase \david{$\agent$} during the examination of the indifference class to which \david{$\matching^*(\agent)$} belongs. By \textbf{D2}, $G^*$ must admit no augmenting path w.r.t.\ \david{$\matching_{\agent-1}$}.
We show, however, that $G^*$ admits an augmenting path starting from \david{$\agent$}. To see this note that, by \textbf{D1} \david{and \textbf{D3}}, and \david{by} the construction of edges $E$, edges \david{$(j,\matching^*(j))$} belong to $G^*$ \david{$\forall j<\agent$}. If \david{$\matching^*(\agent)$} is unmatched in \david{$\matching_{\agent-1}$} then \david{$(\agent,\matching^*(\agent))$} constitutes an augmenting path of size 1 in $G^*$. Otherwise, let \david{${j_1}$} denote the partner of \david{$\matching^*(\agent)$} under \david{$\matching_{\agent-1}$} (note that \david{$j_1 < \agent$)}. It follows from \david{\textbf{D1} and \textbf{D3}}, and the construction of $E$, that \david{${j_1}$} is matched under $\matching^*$. If \david{$\matching^*({j_1})$} is unmatched under \david{$\matching_{\agent-1}$} then we have found an augmenting path of length 3. Otherwise, let \david{${j_2}$} denote the partner of \david{$\matching^*({j_1})$} under \david{$\matching_{\agent-1}$} (note that \david{$j_2 < \agent$)}. The same argument we used for \david{${j_1}$} can be used for \david{${j_2}$}, resulting in either the discovery of an augmenting path of size 5 or reaching a new agent. We can repeatedly use this argument and each time we either find an augmenting path (and stop) or visit an agent that appears \david{in} $\agentsorder$ before \david{$\agent$}.  As each agent is assigned at most one object in every matching, and vice versa, the agents \david{${j_r}$} that we encounter on our search for an augmenting path are all distinct. Therefore, since there are a finite number of agents and objects, we \david{are bound} to reach an object $\object$ \david{that} is unmatched under \david{$\matching_{\agent-1}$}, hence exposing an augmenting path in $G^*$, a contradiction.
\end{proof}

\begin{corollary}\label{cor:par}
The matching produced by SDMT-1 is a Pareto optimal matching.
\end{corollary}
\begin{proof}
By Theorem~\ref{thm:ptm}, SDMT-1 produces an SPM. It follows from Theorems 1 and 2 in \cite{Svensson} that any SPM is a Pareto optimal matching.
Hence, the matching produced by SDMT-1 is a Pareto optimal matching.
\end{proof}

SDMT-1 is truthful, no matter which augmenting path is selected in each phase of the mechanism, as the next result shows.  The proof idea is as follows.  Note that when an agent's turn arrives, SDMT-1 assigns him an object from what the algorithm identifies as his ``best possible indifference class''; i.e., the \david{top-most} indifference class from which he can be assigned an object without harming any previously-arrived agent. Then as soon as he is assigned an object, by Observation \ref{obs:1}, he is guaranteed to be allocated the same object, or one that he equally values, when the algorithm terminates. Hence, as long as we can show that the algorithm correctly identifies these ``best possible indifference classes'', it is straightforward to see that no agent can benefit from misreporting.  The proof of the next theorem formalises this argument.
\begin{theorem}\label{thm:truth}
The mechanism SDMT-1 is truthful.
\end{theorem}
\begin{proof}
Assume, for a contradiction, that the claim does not hold.  Let \david{$\agent$ be the first agent in $\agentsorder$ (i.e., the lowest-indexed agent)}  who benefits from misrepresenting his preferences and reporting \david{$\preflist'(\agent)$} instead of \david{$\preflist(\agent)$}.
Let $\preflist' = \david{(\preflist'(\agent),\preflist(-\agent))}$.

Let $\matching$ denote the matching returned by
SMDT-1 on instance $\inputsetting=(\agentset,\objectset,\preflist)$, i.e., the instance in which agent \david{$\agent$} reports truthfully, and \david{let} \david{$\matching^*$} denote the matching returned on instance  $\inputsetting'=(\agentset,\objectset,\preflist')$. 
\david{Then in $I$, $\matching^*(\agent)\strictlypref_\agent\matching(\agent)$ and $\matching(j) \weaklypref_j \matching^*(j)$, $\forall j < \agent$.}

\david{By Theorem \ref{thm:ptm}, $\matching$ is an SPM in $\inputsetting$, and $\matching^*$ is an SPM in $\inputsetting'$.  Suppose that in $\inputsetting$, $\matching(j) \strictlypref_{j} \matching^*(j)$, for some $j < \agent$.   Let $k$ be the smallest integer such that $\matching(k) \strictlypref_k \matching^*(k)$ in $\inputsetting$.  As $k<\agent$, for each $j$ ($1\leq j\leq k$), agent $j$ has the same preference list in $\inputsetting$ and $\inputsetting'$, by construction of $L'$.  Hence $\matching^*$ cannot be an SPM in $\inputsetting'$ after all, a contradiction.}

\david{It follows that in $I$, $\matching^*(\agent)\strictlypref_\agent\matching(\agent)$ and $\matching(j) \tie_j \matching^*(j)$, $\forall j < \agent$.  We now obtain a contradiction to the fact that $\matching$ is an SPM in $\inputsetting$.}
\end{proof}


We now show a bound on the time complexity of SDMT-1.
Let $\gamma$ denote the size of the largest indifference class for a given instance $\inputsetting$.
\begin{theorem}\label{thm:ptm-runtime}
SDMT-1 terminates in time $O(\numagents^2\gamma\david{+m})$.
\end{theorem}
\begin{proof}
For each agent $\agent$ matched under $\matching$, let $\ell_{\agent}$ denote the length of the indifference class to which $\matching(\agent)$ belongs. Let $|\preflist(\agent)|$ denote the length of agent $\agent$'s preference list, $\forall \agent \in \agentset$. Searching for an augmenting path in a graph $G=(V,E)$ can be done in time $O(|E|)$ using \david{Breadth-First Search (BFS)}. Hence the search for an augmenting path in each phase $\agent$ can be done in time $O(\ell_1+\ell_2+\cdots+\ell_{\agent-1}+|\preflist(\agent)|)$. Therefore SDMT-1 terminates in time \david{$O((\numagents-1)\cdot \ell_{1}+ (\numagents-2)\ell_{2}+\cdots+\ell_{\numagents-1} + \sum_{\agent\in\agentset}|\preflist(\agent)|)$}. However, $\ell_{\agent} \leq \gamma$, $\forall \agent \in \agentset$, therefore $\david{(\numagents-1)\cdot \ell_{1}+ (\numagents-2)\ell_{2}+\cdots+\ell_{\numagents-1} + \sum_{\agent\in\agentset}|\preflist(\agent)|}\leq \numagents^2\gamma + m$, where $m$ is the number of (agent,object) acceptable pairs. Hence SDMT-1 terminates in time $O(\numagents^2\gamma\david{+m})$.
%
\end{proof}

As noted in Section \ref{sec:intro}, in the strict preferences case, any Pareto optimal matching is at least half the size of a maximum size such matching.  The same is true in the general case with indifferences, since any Pareto optimal matching is a maximal matching in the underlying bipartite graph $G_0$ \david{for $\inputsetting$},
and any maximal matching in $G_0$ is at least half the size of a maximum matching in $G_0$ \cite{KH78}.
Hence SDMT-1 obviously achieves approximation ratio $2$ when we are concerned with the cardinality of the matching. We next show that, when agents are assigned arbitrary weights, SDMT-1 achieves the same approximation ratio (relative to a maximum weight Pareto optimal matching) if the agents are ordered in $\agentsorder$ in non-increasing order of their weights.
\begin{theorem}\label{thm:sdmt-apx-ratio}
SDMT-1 achieves approximation ratio of $2$ relative to the size of a maximum weight Pareto optimal matching, if the agents are ordered in $\agentsorder$ in non-increasing order of their weights.
\end{theorem}
\begin{proof}
Given an HA instance $\inputsetting$, let $\matching$ be the matching produced by SDMT-1 and let $\matching'$ be a maximum weight Pareto optimal matching in $\inputsetting$. List the agents matched under each of these matchings in non-increasing order of weight. Let $\agent_1,\ldots, \agent_k$ denote such an order under $\matching$, and let $\agent'_1,\ldots,\agent'_l$ denote such an order under $\matching'$.

Take any agent $\agent'_r$ who is matched under $\matching'$, to say $\object$, but not matched under $\matching$ (if no such agent exists then $\matching$ is itself a maximum weight Pareto optimal matching).
Note that, as $\matching$ is Pareto optimal, $\object$ must be matched under $\matching$, for otherwise $\matching\cup\{(\agent'_r,\object)\}$ Pareto dominates $\matching$. 
As SDMT-1 generates an SPM w.r.t.\ $\agentsorder$ (Theorem~\ref{thm:ptm}) and agents are listed in non-increasing order of weight under $\agentsorder$, it follows that $\object$ must be allocated in $\matching$ to an agent $\agent_s$ who has at least as large a weight as $\agent'_r$ (for otherwise $(\matching\setminus\{(\agent_s,\object)\})\cup\{(\agent'_r,\object)\}$ has a lexicographically smaller signature than $\matching$, a contradiction).

We claim that $\agent_s$ must be matched under $\matching'$ as well, as otherwise $(\matching'\setminus\{(\agent'_r,\object)\})\david{\cup}\{(\agent_s,\object)\}$ has a higher weight than $\matching'$, \david{a contradiction}. (\david{Recall} that a maximum weight Pareto optimal matching must be a maximum weight matching as well by \david{Proposition \ref{prop:weight}}.)
Hence we have established that, for each agent $\agent'_r$ matched under $\matching'$ but not matched under $\matching$, there exists a unique agent $\agent_s$, with weight at least as large as that of $\agent'_r$, who is matched under $\matching$.
Thus if $\agentset_1$ is the set of agents matched in $\matching'$ and $\agentset_2$ is the set of agents matched in $\matching$, it follows that $wt(\agentset_2)\geq wt(\agentset_1\backslash \agentset_2)$, where $wt(\agentset')$ is the sum of the weights of the agents in $\agentset'$, for $\agentset'\subseteq \agentset$.
Also   $wt(\agentset_2) = wt(\agentset_2\backslash \agentset_1)+wt(\agentset_2\cap \agentset_1)\geq wt(\agentset_2\cap \agentset_1) = wt(\agentset_1)-wt(\agentset_1\backslash \agentset_2)\geq wt(\agentset_1)-wt(\agentset_2)$, hence the result.
\end{proof}

It is known (see, e.g., \cite{ACMM04}) that, in the case of strict preferences, not only can we find a Pareto optimal matching using SDM, but we can also generate \emph{all} Pareto optimal matchings by executing SDM on all possible permutations of the agents. In other words, given any Pareto optimal matching $\matching$, there exists an order of the agents such that executing SDM on that order returns $\matching$.  A similar characterisation of Pareto optimal matchings holds in the case of preferences with ties.  This is stated by the following result, whose proof is given in the Appendix.
%
\begin{theorem}\label{sdmt-1-any-PO}
Any Pareto optimal matching can be generated by some execution of SDMT-1.
\end{theorem}

\section{Randomised mechanism with weights and ties}


In this section we will analyse our mechanism for the weighted version of our problem. Our algorithm in the next section is truthful with respect to agents' preferences and weights (under the no-overbidding assumption, \piotr{see also the beginning of Section \ref{sec:bounds}}) and provides an $\frac{e}{e-1}$-approximate Pareto optimal matching.
We will show in Section \ref{sec:bounds} that, even if the weights of all agents are the same our algorithm uses the best possible random strategy -- no other such strategy leads to better than $\frac{e}{e-1}$-approximate matching.

\subsection{Second truthful mechanism: SDMT-2}

\hide{
\begin{theorem}\label{thm:truth}alternating
The mechanism SDMT-1 is truthful.
\end{theorem}
\begin{proof}
We prove this by contradiction, suppose agent $\agent$ changes his preference list from $L(i)$ to $L'(i)$, let $\phi$ be the mechanism SDMT-1. Let $a=\phi_i(L(i),L(-i))$ and $a'=\phi_i(L'(i),L(-i))$ where $a'\strictlypref_i a$. Assume $u=rank(i,a)$ and $u'=rank(i,a')$, where $u'<u$. Let $\mu_l$ ($\mu'_l$, resp.) denotes the matching generated at the end of $l$th iteration by SDMT-1 running on $I=(N,A,L)$ ($I'=(N,A,L')$ where $L'=(L'(i),L(-i))$, resp.), for any $l\in N$. By Theorem~\ref{thm:ptm}, we know that $\mu_l(z)\tie_z\mu'_l(z)$, $\forall 1\le z\le l<i$. Since $a=\phi_i(L(i),L(-i))$, $\mu_i(i)\tie_i a$, similarly, $\mu'_i(i)\tie_i a'$. Then for any $b\in C^i_{u'}$, in the iteration $i$ by SDMT-1 running on $I=(N,A,L)$, there is no augmenting path starting from $(i,b)$, in particular, from $(i,\mu'_i(i))$. Now consider the graph $\mu'_i\oplus\mu_{i-1}$ (the argument is similar to the proof of Theorem~\ref{thm:ptm}), the connected component starting from $(i,\mu'_i(i))$ is an odd length alternating path, which is indeed an augmenting path  starting from $(i,\mu'_i(i))$ in the iteration $i$ by SDMT-1 running on $I=(N,A,L)$, contradiction.
\end{proof}
}
The approximation ratio analysis of the randomised version of SDMT-1 is complex, because it requires additional information which is not maintained by SDMT-1. For the sake of the analysis, we introduce a variant of SDMT-1, called SDMT-2. After introducing some terminology we present SDMT-2, and then establish the equivalence between SDMT-1 and SDMT-2. Pareto optimality and truthfulness of SDMT-2 will then follow from this equivalence and these same two properties of SDMT-1.
We will prove that the randomised version of SDMT-2 is $\frac{e}{e-1}$-approximate. By the equivalence of the two algorithms, a randomised version of SDMT-1 has the same approximation ratio.

Let $o_1\succ o_2\succ\cdots\succ o_{\numobjects}$ be a common order of all the objects. This order will be used to break possible ties in SDMT-2. \piotr{In what follows we will use use lower case letters from the beginning of the alphabet to name individual objects, e.g., $a, b, c, d, e, f, g, h$.} 
\piotr{We define now some notions that will be used to describe algorithm SDMT-2. These definitions will refer to any time point during an execution of this algorithm.}
In the course of the algorithm agents will be (temporarily) {\em allocated} subsets of objects from their preference list. When an agent is allocated a subset of objects we say that he {\em owns} these objects. Let $S\subseteq N$ and suppose that some of the agents in $S$ have been allocated some objects and the allocated objects to each agent appear in the same indifference class of this agent. At any time during the execution of the algorithm, each agent who is allocated more than one object is called {\em labelled} and {\em unlabelled} otherwise. 
\piotr{Likewise, at any point during the execution of the algorithm,} let $i\in N$, and let $B\subseteq L(i)$ be such that $i$ is not \piotr{currently} allocated any object in $B$.
The {\em trading graph (TG)} is a directed graph $TG(i,B,S)$ with $\{i\} \cup S$ as the set of nodes, and arcs defined as follows: Let agent $i$ point to \piotr{each} agent in $S$ who owns any object in $B$. For each unlabelled agent, e.g., $j\in S$, to which $i$ points, suppose the current object allocated to $j$ is in $j$'s $k$th indifference class $C^j_k$. Then let $j$ point to \piotr{each} agent in $S$ who currently owns any object in $C^j_k$ not owned by $j$. Continue this process for the new pointed-to and unlabelled agents until no agent in $S$ needs to point to other agents. %
\hide{Obviously, each agent has at most one time opportunity to point to the other agents and $TG(i,B,S)$ can be constructed in $|B|+|S|\gamma$ time.}See Figure \ref{fig:tradinggraph} for \piotr{an example of} how \piotr{$TG(7,\{g,h\},[6])$} is constructed: \piotr{agent $7$ points to agent $5$ and $6$ since currently agent $5$ owns $g$ and agent $6$ owns $h$; then, as agent $5$ is unlabelled, agent $5$ points to agents $4$ and $1$ since agent $4$ owns $e$ and agent $1$ owns $b$ and $c$; similarly, agent $6$ points to agents $1$ and $3$; agent $3$ points to $1$ and $6$; only agents $1$ and $2$ are labelled.}
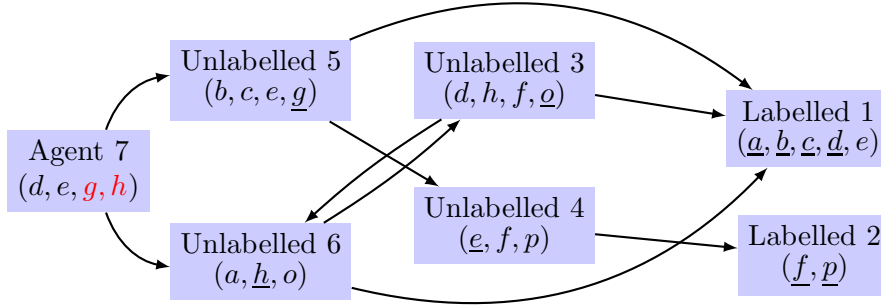
\begin{figure}[t]
	\begin{center}
\begin{tikzpicture}
[scale=.8,auto=left,every node/.style={rectangle,fill=blue!20}]
\tikzset{vertex/.style = {shape=rectangle,draw,minimum size=1.5em}}
\tikzset{edge/.style = {->,> = latex'}}

  \node (v6) at (0,2.5){\shortstack{Agent $7$\\$(d,e,{\color{red}{g,h}})$}};
  \node (v4) at (3,4){\shortstack{Unlabelled $5$\\$(b,c,e,\underline{g})$}};
  \node (v2) at (7,4){\shortstack{Unlabelled $3$\\ $(d,h,f,\underline{o})$}};
  \node (v5) at (3,1){\shortstack{Unlabelled $6$\\ $(a,\underline{h},o)$}};
  \node (v3) at (7,1.6){\shortstack{Unlabelled $4$\\$(\underline{e},f,p)$}};
  \node (v1) at (12,3.2){\shortstack{Labelled $1$ \\$(\underline{a},\underline{b},\underline{c},\underline{d},e)$}};
  \node (v7) at (12.1,1.1){\shortstack{Labelled $2$ \\$(\underline{f},\underline{p})$}};
  
\tikzset{EdgeStyle/.style={->}}
\draw[-latex,thick=30, bend left]  (v6) edge (v4);
\draw[-latex,thick=30, bend right]  (v6) edge (v5);
\draw[-latex,thick=30, bend left]  (v4) edge (v1);
\draw[-latex,thick=30]  (v4) edge (v3);
\draw[-latex,thick=30]  (v3) edge (v7);
\draw[-latex,thick=30]  (v2) edge (v1);
\draw[-latex,thick=30, bend right]  (v5) edge (v1);

\draw[-latex,thick=30, bend right=5]  (v2) edge (v5);
\draw[-latex,thick=30, bend right=5]  (v5) edge (v2);

\end{tikzpicture}
\end{center}
	\vspace{0.5cm}
	\caption{The trading graph \piotr{$TG(7,\{g,h\},[6])$}, $\underline{h}$ denotes $h$  is owned currently by the agent. \piotr{Objects in the parentheses below each agent represent a single indifference class of that agent. The common order of the objects is $a\succ b\succ c\succ d\succ e\succ f\succ g\succ h\succ o \succ p$ (used in the text below).}}
	\label{fig:tradinggraph}
\end{figure}

Let $H=\left\{ a\in L(i)\,|\,  \text{there is a (directed) path from}\, i \, \text{to a labelled agent in}\,\,  TG(i,a,S)\right\}.$ 
Note that, as labelled agents do not point to any agents, no intermediate agent on a directed path is labelled.
Note that $H$ may be empty, and it can be found,\hide{in $|L(i)|+|S|\gamma$ time} \piotr{for instance,}\footnote{\piotr{Here, what only matters is the reachability, that is, existence of such directed path in $TG(i,a,S)$ from agent $i$ to a labelled agent.}} by breadth first search (BFS). If $H\neq \emptyset$, let $\ell$ be the highest indifference class of $i$ with $H\cap C^i_{\ell}\neq \emptyset$. Define $\max TG(i,L(i),S)$ to be the highest order object in $H\cap C^i_{\ell}$ (e.g., in Figure~\ref{fig:tradinggraph}, \piotr{$\max TG(7,\{d,e,g,h\},[6])=g$}). \piotr{We also explicitly define $\max TG(i,L(i),S) =\emptyset$ if $H = \emptyset$.}\hide{since there is a path from $6$ to the labelled agent $1$ in either $TG(6,\{g\},[5])$ or $TG(6,\{h\},[5])$, no path from $6$ to labelled agent in $TG(6,\{d\},[5])$ and $TG(6,\{e\},[5])$, as we know $g\succ h$  in the common order, hence, $\max TG(6,\{d,e,g,h\},[5])=\{g\}$).}~If $\max TG(i,L(i),S)\neq \emptyset$, then there is a path from $i$  to a labelled agent in $TG(i,a,S)$, which can be found by BFS. \hide{(note that  only the end point of the path is labelled).} Suppose the path is $(i_0,i_1,i_2,\cdots,i_k)$, where $i_0=i$ and only $i_k$ is labelled. 
Now denote $Trading(i,a,S)$ to be a procedure that allocates the object owned by $i_{s+1}$ to $i_s$, for $s=0,1,\cdots,k-1$. Note that $i_k$ may own more than one object for which $i_{k-1}$ has pointed to $i_k$. In this case, the highest order object among such objects is allocated to $i_{k-1}$.
After trading, if $i_k$ still owns more than one object, keep $i_k$ labelled and unlabel $i_k$ otherwise. In Figure~\ref{fig:tradinggraph}, \piotr{considering procedure $Trading(7,g,[6])$, we note that there are two paths from agent $7$ to a labelled agent: $(7,5,1)$ and $(7,5,4,2)$.} \piotr{Procedure $Trading(7,g,[6])$ can use any of those two paths. If 
$Trading(7,g,[6])$ uses the first path, then it
allocates $g$ to agent $7$ and $b$ to agent $5$, since $b\succ c$, and keeps agent $1$ labelled. If procedure
$Trading(7,g,[6])$ uses the second path, then it
allocates $g$ to agent $7$ and $e$ to agent $5$, anf $f$ to agent $4$, since $f\succ p$, and changes agent $2$ to unlabelled.}

\piotr{Recall} that $C^i_{\numobjects+1}=\emptyset$, $\forall i\in [\numagents]$. With these preliminaries, we present our algorithm SDMT-2 (see Algorithm \ref{alg:indifference}). In the following, we will refer to $k$th iteration of the ``for loop'' in SDMT-2 as \piotr{the} $k$th loop. Observe that in the $k$th loop, $j_1$ is the highest indifference class of $i$ where $i$ can obtain unallocated objects, and $j_2$ is the highest indifference class of $i$ where $i$ can obtain objects from the allocated objects without hurting the agents prior to $i$.
\begin{algorithm}[t!]\label{alg:indifference}
\SetAlgoLined
\caption {Serial Dictatorship Mechanism with Ties, version 2 (SDMT-2)}
\KwIn{ Agents $N$; Objects $\objectset$; Preference list profile $L$; An order of agents $\agentsorder$, w.l.o.g. let $\sigma(i)=i$, $\forall i\in N$ }
\KwOut{Matching }
Let $\objectset_1\leftarrow \objectset$ \hspace*{3mm} // $\objectset_1$ is the set of currently unallocated objects \\
\For {each agent $i\in \agentset$ in the order of $\agentsorder$}{
Define
$ j_1=\left\{\begin{array}{ll}\min\{j:\objectset_i\cap C^i_j\neq\emptyset\}& \text{if}\, \objectset_i\cap L(i)\neq \emptyset;\\ \numobjects+1 & \text{otherwise.}
 \end{array}\right.$
$j_2=\left\{\begin{array}{ll}\min\{j:\max TG(i,L(i),[i-1])\in C^i_{j}\}& \text{if}\, \max TG(i,L(i),[i-1])\neq \emptyset;\\ \numobjects+1 & \text{otherwise.}
 \end{array}\right.$\\
 \If {$j_1\le j_2$} {Allocate all the objects in $\objectset_i\cap C^i_{j_1}$ to $i$; Label $i$ if $|\objectset_i\cap C^i_{j_1}|\ge 2$;
 $\objectset_{i+1}\leftarrow \objectset_{i}\backslash (\objectset_i\cap C^i_{j_1})$}

  \Else { $Trading(i,\max TG(i,L(i),[i-1]),[i-1])$;
  $\objectset_{i+1}\leftarrow \objectset_{i}$ }
}
For each labelled agent, allocate to him the highest order object he currently owns.\\
For each unlabelled agent, if he currently owns an object, allocate it to him.\\
Output the matching.
\end{algorithm}

\hide{By the process of SDMT-2, we have the following simple but  useful observation:}
\begin{observation}\label{obs:ind1}
For each agent $\agent$,
after $\agent$'s turn in ``for loop'' of SDMT-2,
if $\agent$ is allocated no object, then he will be allocated no object when SDMT-2 terminates. Otherwise, if $\agent$ is provisionally allocated some objects in his turn, then in the final matching he will be allocated an object in the same indifference class as his initially allocated objects.
\end{observation}

\begin{observation}\label{obs:ind2}
For each agent $\agent$, 
after $\agent$'s turn,
if $\agent$ is allocated an object $\object\in C^{\agent}_j$, then all the objects in $\cup_{k=1}^j C^{\agent}_k$ have been allocated to either $i$ or to some agents prior to $i$. Once an object is allocated, it remains allocated until the end of the for loop.
\end{observation}

Now we establish the equivalence of SDMT-1 and SDMT-2.
\begin{theorem}\label{thm:equ}
Given the same input, SDMT-1 and SDMT-2 match the same set of agents. Furthermore, for each matched agent $\agent$, the object allocated to $\agent$ in SDMT-1 is in the same indifference class of $i$ as the object allocated to him in SDMT-2. \piotr{This equivalence between SDMT-1 and SDMT-2 holds for any fixed common order $\succ$ of the objects used in SDMT-2 and it is also independent of how SDMT-2 finds the directed paths from agent $i$ to a labelled agent in the trading graph $TG(i,L(i),[i-1])$.}
\end{theorem}
\begin{proof}
We will prove the following two facts inductively which simply implies the conclusion of Theorem \ref{thm:equ}. \piotr{Without loss of generality,} suppose the order of agents is $\agentsorder(i)=i$, $ \forall i\in N$. Until the step $i$,
\begin{itemize}
\item[1.] for each agent $k \leq i-1$, the allocated objects of SDMT-1 and SDMT-2 to $k$ are in the same indifference class, (if one of them is empty, the other is empty as well)
\item[2.] for each $\ell\le \numobjects$, and $a\in C^i_{\ell}$,  there is an augmenting path starting from $(i,a)$ in SDMT-1 if and only if $a$ is unallocated in SDMT-2 or there is a path from $i$ to a labelled agent in $TG(i,a,[i-1])$ in SDMT-2.
\end{itemize}
Consider the base case, for agent $1$, obviously property $1$ is true since they are all empty sets. \piotr{Now, for property 2, let $\ell\le \numobjects$ and $a\in C^1_{\ell}$}, then there is an augmenting path from \piotr{$(1,a)$} in SDMT-1 and $a$ is unallocated in SDMT-2. \piotr{This shows both implications of property 2 for agent 1.}\\
\piotr{For the proof of the induction step,} suppose \piotr{properties} 1 and 2 \piotr{are} true for all the steps $k \le i-1$, we now prove that they are true for step $i$. For property 1, by inductive hypothesis, property 1 holds for any $k \leq i-2$. \piotr{Since} property 2 holds for agent $i-1$ by inductive hypothesis,  the objects allocated to agent $i-1$ in SDMT-1 and SDMT-2 will be in the same indifference class, thus, property 1 holds for step $i$. Now property 2 will be proved true for agent $i$, for each $\ell\le \numobjects$, and $a\in C^i_{\ell}$:

 For $\Rightarrow$ direction, if  there is an augmenting path starting from $(i,a)$ in SDMT-1,  and if $a$ is allocated previously in SDMT-2, suppose the new matching generated in SDMT-1 due to the augmenting path is $(k,\mu(k))$, $k\le i$, where $\mu(i)=a$. By property 1 of inductive hypothesis and Observation \ref{obs:ind2}, all the objects in $\{\mu(k),k\le i\}$ have been allocated to some agents $k \leq i-1$ in SDMT-2. For object $b$, we use $\nu^{-1}(b)$ to denote the agent whom $b$ is allocated to in SDMT-2. Now consider the following path in $TG(i,a,[i-1])$: let $i_1=\nu^{-1}(\mu(i))$, \piotr{and} if $i_1$ is labelled then we are done, otherwise, let  $i_2=\nu^{-1}(\mu(i_1))$. If $i_2$ is labelled, then we are done, otherwise continue this process. Finally, \piotr{we will reach by this process a} labelled agent \piotr{among the agents in $[i-1]$. This is true because of the  pigeonhole principle: $i$ objects from $\{\mu(k), k \leq i\}$ are allocated in SDMT-2 to $i-1$ agents in $[i-1]$.}
 

For $\Leftarrow$ direction, now suppose $a$ is unallocated or there is a path from $i$ to a labelled agent in $TG(i,a,[i-1])$ in SDMT-2. Suppose $a$ is allocated and there is a path from $i$ to a labelled agent in $TG(i,a,[i-1])$. Then by $Trading(i,a,[i-1])$, we can make all the agents $k\le i$ allocated at least one object and $i$ is allocated $a$. \piotr{This defines an allocation of (sets of) objects to agents $k\le i$ in SDMT-2.} Let us now select any matching \piotr{using this allocation}, e.g., $M=\{(k,\nu(k)), k\le i\}$, where $\nu(i)=a$ (we can also select such a matching if $a$ is unallocated in SDMT-2). \piotr{For instance, matching $\nu$ can assign the hightest order object to each agent $k\le i-1$ from the current set of objects allocated to $k$, and assign object $a$ to agent $i$.} Suppose the matching generated after step $i-1$ in SDMT-1 is $M'=\{(k,\mu(k)),k\le i-1\}$. By property 1 of inductive hypothesis, we know $\mu(k)$ and $\nu(k)$ are in the same indifference class of agent $k$, for any $k\in [i-1]$. Now consider $M \oplus M'$, which consists of alternating paths and cycles. Then a connected component of $M \oplus M'$ that contains $(i,\nu(i))$ must be an odd length alternating path in $M \oplus M'$ w.r.t. $M'$, implying an augmenting path starting from $(i,a)$ in SDMT-1. The argument showing that the connected component that contains $(i,\nu(i))$ must be an odd length alternating path is as follows. If $\nu(i)=a$ is unallocated in SDMT-1, then $(i,\nu(i))$ is an odd length alternating path.
\piotr{Otherwise,} suppose $i_1=\mu^{-1}(\nu(i))$, then consider whether $\nu(i_1)$ is allocated or not in SDMT-1. If not we get an odd path $(i,\nu(i),i_1,\nu(i_1))$. Otherwise continue the search, and let $i_2=\mu^{-1}(\nu(i_1))$, then consider whether $\nu(i_2)$ is allocated or not in SDMT-1. If not we get an odd length path $(i,\nu(i),i_1,\nu(i_1),i_2,\nu(i_2))$, and so on. Finally, we will get an odd length alternating path starting from $(i,\nu(i))=(i,a)$ w.r.t. $M'$, which is indeed an augmenting path starting from $(i,a)$ in SDMT-1. This concludes the \piotr{proof of the} induction \piotr{step}.
\end{proof}

It is easy to see that both SDMT-1 and SDMT-2 reduce to SDM if all agents have strict preference over objects.

\begin{theorem}\label{thm:time}
SDMT-2 is truthful, Pareto optimal, and terminates in \david{$O(\numagents^2\gamma+m)$} running time.
\end{theorem}
\begin{proof}
The first two properties follow from the equivalence between SDMT-1 and SDMT-2 (Theorem \ref{thm:equ}) and the Pareto optimality (Corollary~\ref{cor:par}) and truthfulness (Theorem~\ref{thm:truth}) of SDMT-1. 
It remains to establish the running time of SDMT-2.

\piotr{By the previous analysis given in the proof of Theorem \ref{thm:ptm-runtime}, in each loop iteration $i$, the running time is $O(|L(i)|+(i-1)\gamma)$.  Summing $i$ over $[\numagents]$, we obtain that the running time of SDMT-2 is $O(m+\numagents^2\gamma)$}.
\end{proof}

\subsection{Randomised mechanism}
We now present a universally truthful and Pareto optimal mechanism with approximation ratio of $\frac{e}{e-1}$, where agents may have weights and their preferences may involve ties (see Algorithm \ref{alg:randomindifference}, where \piotr{$e^{Y_i-1} = g(Y_i)$}). \piotr{Note that in the absence of agents' weights, sorting agents in the decreasing order of $w_i(1-g(Y_i)$ simply means to sort them in the increasing order of the $Y_i$ values, so the exponentiation is only used for the correct handling of the weights.}

When preferences are strict, Algorithm \ref{alg:randomindifference} reduces to a variant of RSDM that has been used in weighted online bipartite matching with approximation ratio $\frac{e}{e-1}$ (see \cite{Aggarwal-etal} and \cite{Devanur-etal}).
Our analysis of Algorithm \ref{alg:randomindifference} is a non-trivial extension of the primal-dual analysis from \cite{Devanur-etal} to the case where agents' preferences may involve ties. \piotr{Before analysing the approximation ratio, we will argue about the universal truthfulness of Algorithm \ref{alg:randomindifference} when agents' preferences are private and they in addition have weights. \hide{We note here, however, that the truthfulness with private weights is not the focus of our paper.}}

\begin{algorithm}[t!]\label{alg:randomindifference}
\SetAlgoLined
\caption {Random SDMT-2 for Weighted Agents with Ties}
\KwIn{ Agents $N$; Objects \piotr{$\objectset$}; Preference list profile $L$; Weights $W$}
\KwOut{Matching }
\For {each agent $i\in N$}{
Pick $Y_i\in[0,1]$ uniformly at random\;
}
Sort agents in decreasing order of $w_i(1-e^{Y_i-1})$ (break ties in favour of smaller index)\;
Run SDMT-2 according to above order\;
Return the matching\;
\end{algorithm}

   If the weights are public, Algorithm \ref{alg:randomindifference} is universally truthful and Pareto optimal. This is because it chooses a random order of the agents, 
given the weights, 
and then runs SDMT-2 according to this order. It follows by inspection of SDMT-2 that, if the order of the other agents is given, an agent can get a better object if he appears earlier in this order. Then it is not difficult to see that if the weights are private, and under the assumption that no agent is allowed to bid over his private weight \piotr{(the so-called no-overbidding assumption -- see the beginning of Section \ref{sec:bounds})}, Algorithm \ref{alg:randomindifference} is still universally truthful in the sense that no agent will lie about his preferences or weight.
\begin{theorem}\label{thm:truthwithweight}
Algorithm \ref{alg:randomindifference} is universally truthful, 
even if the weights and preference lists of the agents are their private knowledge, assuming that no agent can over-bid his weight.
\end{theorem}
\begin{proof}
Algorithm \ref{alg:randomindifference} is a distribution \piotr{over} deterministic mechanisms due to the selection of random variables $Y_i$. For each deterministic mechanism (i.e.,  SDMT-2 when $Y_i$, $i\in N$ is fixed), we prove \piotr{that} it is truthful \piotr{with respect to} weights and preference lists. \piotr{Let us denote} by $\phi$ the mechanism of SDMT-2 when $Y_i$, $i\in N$ is fixed. If is not difficult to see that for any $(W,L)$, $w'_i\le w_i$ and $L'(-i)$, $i\in N$, we have $\phi_i(W,L)\weaklypref_i\phi_i((w'_i,w_{-i}),L)\weaklypref_i\phi_i((w'_i,w_{-i}),(L'(i),L(-i)))$. The first preferred order in this chain follows from the fact that the order of $i$ when $i$ bids $w_i$ is better than or equal to his order when he bids $w'_i$. The second preferred order in this chain follows by the truthfulness of SDMT-2 when weights are public.
\end{proof}

\section{Analysis of the approximation ratio}

\piotr{To gain some high-level intuition behind our extension from strict preferences to preferences with ties, we highlight here the similarities and differences between our problem and that of online bipartite matching. Our problem with strict preferences and without weights is closely related to online bipartite matching.\footnote{In the online bipartite matching problem \cite{Bhalgat-etal}, vertices of one partition (think of them as agents) are given and fixed, while vertices of the other partition (think of them as objects) arrive in an adversarial order. When an item arrives, we get to see the incident edges on agents. These edges indicate the set
of agents that desire this object. The algorithm must immediately match this object to one of the
unmatched agents desiring it (or choose to throw it away). In the end, the size of the obtained
matching is compared with the optimum matching in the realised graph.} 
If each agent in our problem ranks his desired objects in the order that 
precisely follows the arrival order of objects in the online bipartite matching, the two problems are equivalent. Therefore, we extend the analysis of this particular
setting, where each agent's preference list is a sublist of a global preference list,
to the general case where agents preferences are not constrained and may involve ties, and furthermore agents may have weights.}

To analyze the approximation ratio of Algorithm \ref{alg:randomindifference}, we first write the LP formulation of the (relaxed) problem and its dual LP formulation.
Given random variables $Y_i$, we will define a primal solution and a dual solution obtained by Algorithm \ref{alg:randomindifference}, which are both random variables, such that the objective value of the primal solution is always at least a fraction $F$ of the objective value of the dual solution, and that the expectation of duals is feasible. Hence, the expectation of the primal solution is at least $F$ times the expectation of duals, which by weak LP duality, is at least \piotr{the} optimal value of the primal LP.  We now give the standard LP and its dual of our problem.  In what follows, $G=(V,E)$, where $V=N\cup \objectset$ and $E=\{(i,a),i\in N,a\in  \objectset\}$.
\[
\begin{array}{llll}
\multicolumn{2}{l}{\max \sum_{(i, a)\in E}w_ix_{ia} ~~ \mbox{such that}} \qquad \qquad \qquad & 
                                                                         \multicolumn{2}{l}{\min \sum_{i\in N}\alpha_i +\sum_{a\in \objectset}\beta_{a} ~~ \mbox{such that}} \vspace{2mm}\\
~~ & \forall i\in N :  \sum_{a:(i,a)\in E}x_{ia}\le 1                    & ~~ & \forall (i,a)\in E : \alpha_i +\beta_{a}\ge w_i \vspace{2mm}\\
~~ & \forall a\in \objectset : \sum_{i:(i,a)\in E}x_{ia}\le 1                     & ~~ & \forall i\in N : \alpha_i\ge 0 \vspace{2mm}\\
~~ & \forall (i, a)\in E : x_{ia}\ge 0                                   & ~~ & \forall a\in \objectset : \beta_{a}\ge 0
\end{array}
\]
By the next result, proved in \cite{Devanur-etal}, the inverse of approximation ratio is $F\in[0,1]$.
\begin{lemma}[\cite{Devanur-etal}]\label{lem:factor1}
Suppose that a randomised primal-dual algorithm has a primal feasible solution with value $P$ (which is a random variable) and a dual solution which is not necessarily feasible, with value $D$ (which is also a random variable) such that
\begin{itemize}
\item[1.] for some universal constant $F$, $P\ge F\cdot D$, always, and
\item[2.] the expectation of the randomised dual variables forms a feasible dual solution, that is, $\E(\alpha_i)$ and  $\E(\beta_a)$ are dual feasible.
\end{itemize}
The expectation of $P$ is then at least $F\cdot$ OPT where OPT is the value of the optimum solution.
\end{lemma}

\begin{proof}
Since $P\ge F\cdot D$, taking expectations, $\E(P)\ge F\cdot\E(D)$. The cost of the dual solution obtained by taking expectations of the dual \piotr{random} variables is $\E(D)$ and they form a feasible dual solution, therefore $\E(D)\ge$ OPT. Hence, $\E(P)\ge F\cdot$OPT.
\end{proof}

Note that in Lemma \ref{lem:factor1}, OPT is the \david{weight} of maximum \david{weight} matching, which is \david{equal to} the \david{weight} of \david{a} maximum \david{weight} Pareto optimal matching \david{by Proposition \ref{prop:weight}}. Hence, if the condition of Lemma \ref{lem:factor1} holds, the approximation ratio of the mechanism is $\frac{1}{F}$.
The construction of the duals depends on function $g$. Let $F=(1-\frac{1}{e})$.
For any random selection of $Y_i$, $i\in N$, let $\vec{Y}=(Y_1,Y_2,\cdots,Y_{\numagents})=(Y_i,Y_{-i})$. Following the procedure of Algorithm \ref{alg:randomindifference}, whenever agent $i$ is matched to object $a$, let $$x_{ia}(\vec{Y})=1,\quad \alpha_i(\vec{Y})=w_ig(Y_i)/F,\quad \beta_a(\vec{Y})=w_i(1-g(Y_i))/F.$$ 
For all unmatched $i$ and $a$, set $x_{ia}(\vec{Y})=\alpha_i(\vec{Y})=\beta_a(\vec{Y})=0$. By this definition, it follows that for any $Y_i$, $i\in N$, the random value $P$ of the primal solution $\{x_{ia}(\vec{Y}), i\in N, a\in A\}$ is always \piotr{identical to} $F\cdot D$, where $D$ is the random value of the dual solution $\{\alpha_i(\vec{Y}), i\in N, \beta_a(\vec{Y}), a\in \objectset\}$.

Hence, to satisfy the conditions of Lemma \ref{lem:factor1}, we need to show that the expectation of the dual solution $\{\alpha_i(\vec{Y}), i\in N, \beta_a(\vec{Y}), a\in \objectset\}$ is feasible for the dual LP, implying that the approximation ratio of Algorithm \ref{alg:randomindifference} is at most $\frac{1}{F}=\frac{e}{e-1}$.
The main technical difficulty lies in proving the dominance lemma and the monotonicity lemma (see Lemma \ref{lem:dominanceindifference} and  \ref{lem:monotonicityindifference}). 
To prove these two lemmas, \piotr{for any fixed agent $i$, and any fixed object $a \in \objectset$, we define a threshold, denoted by $\theta=\theta^i_a$, of the random variable for $Y_i$, which specifies whether agent $i$ will get matched -- see Lemma \ref{lem:dominanceindifference}.} This threshold will depend on the other agents $Y_{i-}$.
For an agent with strict preferences, such a threshold is the same as that defined in the online bipartite matching problem. However, in the presence of ties, the same defined threshold does not work. We show how to define such a threshold for our algorithm.

Let us fix an agent $\agent\in [\numagents]$ and object $a \in \objectset$, such that $(\agent,a)\in E$. Also, we fix $Y_{-\agent}$, that is, the random variables $Y_{\agent'}$ for all other agents $\agent' \not = i$. We use $\sigma$ to denote the order of agents under $Y_{-\agent}$, i.e., $\sigma(1)$ is the first agent, and so on, and $\sigma([\agent])=\{\sigma(1),\sigma(2),\cdots,\sigma(\agent)\}$. 
Consider Algorithm \ref{alg:randomindifference} running on the instance without agent $\agent$ and let us denote this procedure by $ALG_{-\agent}$, where $\sigma$ is the order of agents under $ALG_{-\agent}$. \piotr{Given agent $\agent$ and object $a$,} the threshold $\theta = \piotr{ \theta^i_a}$ is then defined as follows:
\begin{enumerate}
\vspace*{-1mm}
  \item If $a$ is unmatched in $ALG_{-i}$, let $\theta=1$.
  \item Otherwise, suppose that $a$ is matched in $ALG_{-i}$ to some agent $i'$. Then consider the allocations just
            after the ``for loop'' in SDMT-2 within $ALG_{-i}$ terminated. \\
            If $i'$ is labelled, set $\theta=1$.
  \item Otherwise, suppose $a\in C^{i'}_j$ and construct the trading graph
            $TG(i',C^{i'}_j\backslash\{a\},[\numagents]\backslash\{i\})$ from all the objects in $C^{i'}_j$ other than $a$ (note that
            $\sigma([\numagents-1])=[\numagents]\backslash \{i\}$). Recall that graph $TG(i',C^{i'}_j\backslash\{a\},[\numagents]\backslash\{i\})$ contains directed paths to all agents who can potentially provide an object for $i'$ to trade without affecting any other agent.\\
            If there is a path in $TG(i',C^{i'}_j\backslash\{a\},[\numagents]\backslash\{i\})$ from $i'$ to a labelled agent, set $\theta=1$.
  \item  Otherwise, define\\
           \hspace*{-4mm} $i''= arg\min_{\ell} \{w_{\ell}(1-g(Y_{\ell}))| \text{ there is a path from}\,\, i' \,\, \text{to}\,\, \ell \,\, \text{in} \,\, TG(i',C^{i'}_j\backslash\{a\},[\numagents]\backslash\{i\}) \}$ \\
            Note: If index $\ell$ with minimum value of $w_{\ell}(1-g(Y_{\ell}))$ is not unique, we take for $i''$ the largest such index. Also, observe that either $i'=i''$ or agent $i'$ is before $i''$ with respect to order $\sigma$. \\
            If $w_i(1-g(y))=w_{i''}(1-g(Y_{i''}))$ has a solution $y \in [0,1]$ define $\theta$ to be this solution. \\
            ($g(y)$ is strictly increasing so if there is a solution, it is unique)
   \item Otherwise define $\theta$ to be $0$.
\end{enumerate}
\vspace*{-1mm}
Now consider Algorithm \ref{alg:randomindifference} running on the original instance (denote such procedure as $ALG$), with $(Y_i, Y_{-i})$ fixed. Suppose that $\tau$ is the order of agents under this execution of $ALG$.
The intuition behind the definition of $\theta$ is the following. Having $Y_{-i}$ fixed, we want to define $\theta$ such that if we run $ALG$ with $(Y_i, Y_{-i})$ where $Y_i < \theta$, then agent $i$ gets matched. If 1.~holds, then $Y_i < 1$ and $i$ will be matched because at least object $a$ is his available candidate. If 2.~happens, then $Y_i < 1$ and $i$ will also be matched because object $a$ can be re-allocated from the labelled agent $i'$ to $i$. Case 3.~is analogous to 2.~with the only difference that we now have a trading path from $i'$ to a labelled agent. Finally, case 4.~will be discussed just after Observation \ref{obs:ind3}.

\piotr{In our further analysis, we will need the following notion of a \emph{frozen} agent or object.}
\piotr{
\begin{definition}\label{def-frozen}
We say an agent (respectively, an object) is {\em frozen} if the allocation of this agent (respectively, object) 
remains the same until the termination of SDMT-2. We also say a trading graph is frozen if all of its agents are frozen.
\end{definition}
}

\begin{observation}\label{obs:ind3}
Assume that agent $\agent$ is unmatched in his turn in the ``for loop'' of $ALG$. Suppose $\tau(u)=\agent$, which means $\agent$ selects his object in $u$-th iteration of the ``for loop''. Then at the end of the $k$-th iteration of the ``for loop'', for every $k\ge u$, there is no path from $\agent$ to a labelled agent in $TG(\agent,L(\agent),\tau([k]))$, \piotr{meaning this graph is frozen.}
\end{observation}
\begin{proof}
By SDMT-2, we know \piotr{that} $\piotr{\objectset}_u\cap L(\tau(u))=\piotr{\objectset}_u\cap L(i)=\emptyset$, which means that all the objects in $L(i)$ have been allocated to agents $\tau([u-1])$. Since  $\tau(u)$ is unmatched, there is no path from $\tau(u)$ to any labelled agent in $TG(\tau(u),L(\tau(u)),\tau([u-1]))$.  Let $S\subseteq \tau([u-1])$ be the set \piotr{that is} reachable from $\tau(u)$ in $TG(\tau(u),L(\tau(u)),\tau([u-1]))$. Clearly, each agent in $S$ is unlabelled. Actually, notice that any agent in $S$ is frozen. Therefore, any path through $i$ after $u$-th iteration will reach an unlabelled agent.
\end{proof}

The following two properties (dominance and monotonicity) are well known for agents with strict preference orderings. We generalise them to agents with indifferences. The difficulty of proving both dominance and monotonicity lemmas (Lemma \ref{lem:dominanceindifference} and \ref{lem:monotonicityindifference}) lies in case 4.~(in the definition of threshold $\theta$). This is our main technical contribution as compared to the analysis in \cite{Devanur-etal}.

 Recall that $\tau$ ($\sigma$, resp.) is the order of agents under the execution of $ALG$ ($ALG_{-i}$, resp.). We first discuss intuitions behind case 4.~in \piotr{the} context of the \piotr{Dominance Lemma (Lemma \ref{lem:dominanceindifference})}. Note that in this case there is a path from $i'$ to $i''$ in $TG(i',C^{i'}_j\backslash\{a\},[\numagents]\backslash\{i\})$ and agent $i''$ is unlabelled. We will prove the Dominance Lemma by contradiction, using the following two main steps. \piotr{Indeed, let us assume towards a contradiction, see the text of Lemma \ref{lem:dominanceindifference}, that 
$Y_i < \theta$ and $i$ is not matched in $ALG$.} Then the outcome of $ALG$ is the same as that of $ALG_{-i}$ for all the other agents \piotr{(except agent $i$).} Suppose $\sigma(u)=i''$ in $ALG_{-i}$, then $\tau(u+1)=i''$ in $ALG$ under case 4.
\piotr{Based on the fact that outcomes of $ALG$ and $ALG_{-i}$ are the same (for all agents except agent $i$),} first, we prove that either $i'$ is labelled or there is a path, let us call it $P_1$, from $i'$ to a labelled agent in $TG(i',C^{i'}_j,\tau([u]))$ at the end of the $u$-th iteration of the ``for loop'' in $ALG$. Secondly, due to the above property, we argue that there is a path, let us call it $P_2$, from $i$ to a labelled agent in $TG(i,a,\tau([u]))$ at the end of the $u$-th iteration of the ``for loop'' in $ALG$, contradicting Observation \ref{obs:ind3}; thus $i$ will be matched. Path $P_2$ is constructed by the concatenation of arc $(i, i')$ and path $P_1$, or the concatenation of arc $(i,i''')$, for some $i'''$ on path $P_1$, and the rest of path $P_1$. The existence of $P_1$ is proved by a careful analysis of the structure of frozen subgraphs of the trading graph as the algorithm proceeds; the details can be found in the proof of Lemma \ref{lem:dominanceindifference}.

\begin{lemma}[Dominance Lemma]\label{lem:dominanceindifference}
Given $Y_{-i}$, $i$ gets matched \piotr{(to some object)} if $Y_i<\theta$.
\end{lemma}
\begin{proof} \piotr{Let us assume towards a contradiction that $Y_i < \theta$ and $i$ is not matched in $ALG$. We will consider the following cases below.}

\noindent
{\bf Case 1.} \piotr{(Corresponding to case 1 in the definition of threshold $\theta$.)} If $a$ is unmatched in $ALG_{-i}$, then $\theta=1$. Suppose  agent $i$ is unmatched in $ALG$, then procedure $ALG$ is the same as $ALG_{-i}$ for all the other agents except $i$. But then $a$ is always available to agent $i$, meaning $a$ will be matched to agent $i$ by process of SDMT-2, contradiction.\\
{\bf Case 2.} If $a$ is matched to $i'$ in $ALG_{-i}$: \\
\piotr{{\bf Case 2-(i).}} \piotr{(Corresponding to cases 2 and 3 in the definition of threshold $\theta$.)} If $i'$ is labelled or if there is a path from $i'$ to a labelled agent in $TG(i',C^{i'}_j\backslash\{a\}, [\numagents]\backslash\{i\})$, and $i$ is unmatched, then there is a path from $i$ to a labelled agent in $TG(i,a,[\numagents])$. In this case, by $Trading(i,a,[\numagents])$, we obtain a Pareto improvement, contradicting that SDMT-2 is Pareto optimal.\\
\piotr{{\bf Case 2-(ii).}} \piotr{(Corresponding to cases 5 and 4 in the definition of threshold $\theta$.)} The case $\theta=0$ is trivial, so we consider that $w_i(1-g(y))=w_{i''}(1-g(Y_{i''}))$ has a solution. Suppose that $\sigma(u)=i''$ in $ALG_{-i}$, then if  $Y_i<\theta$, we know that $w_i(1-g(Y_i))>w_{i''}(1-g(Y_{i''}))$, meaning the agent $i$ is prior to agent $i''$ in $ALG$. Then $\tau(u+1)=i''$ in $ALG$.  If $i$ is unmatched in $ALG$, then procedure $ALG$ is the same as $ALG_{-i}$ for all the other agents except $i$.

Suppose $i''$ is allocated an object $b$ in $ALG$. If $i''=i'$, then $b=a$, and if in addition $a\in \objectset_{u+1}$, this means $a$ is always available to all the agents prior to $\tau(u+1)=i''=i'$. Therefore, $a$ will be available to $i$ when $i$ initially selects objects, implying that $i$ must be allocated to some object in his turn, \piotr{leading to a} contradiction.

The case $a\notin \objectset_{u+1}$ is \piotr{analyzed similarly to} the case $i''\neq i'$, so we consider that  $i''\neq i'$. Since there is a path from $i'$ to $i''$  after the ``for loop'' in $ALG$ terminates, $i$ is still unmatched because of our assumption \piotr{towards a contradiction.} Suppose that in this path $\tau(k)$ points to $\tau(u+1) = i''$, for some $k\le u$, then $b$ is available to $\tau(k)$ or $b$ has been allocated before \piotr{the} $k$-th iteration of the ``for loop'' in $ALG$. Since finally $\tau(k)$ gets an object in the same indifference class as $b$ by Observation \ref{obs:ind1}, before \piotr{the} $(u+1)$-st iteration of the ``for loop'' in $ALG$, $b$ has been allocated by Observation \ref{obs:ind2}. Hence, in the $(u+1)$-st iteration of the ``for loop'' in $ALG$, $\tau(u+1)$ gets object $b$ through the trading graph.
\begin{observation}\label{o:same_traging_graph} The trading graph $TG(i',C^{i'}_j\backslash\{a\},[\numagents])$ after the ``for loop'' in $ALG$ terminates, is exactly the same as $TG(i',C^{i'}_j\backslash\{a\},\tau([u+1]))$ at the end of the $(u+1)$-st iteration.
\end{observation}
\piotr{This observation follows from the fact that} otherwise, some agent $\tau(\ell)$ may be reachable from $i'$, where $\ell>u+1$, by process of SDMT-2, contradicting the definition of $i''$; note that we used here the largest index tie breaking rule.

 Therefore, at the end of the $u$-th iteration of the ``for loop'' of $ALG$, suppose that $B$ is the set of objects allocated to $i'$. Then we have the following three cases (note that $ALG$ is the same as $ALG_{-i}$ for all the other agents except $i$\david{)}:
\begin{itemize}
\item[] \vspace*{-2mm} \piotr{{\bf Case 2-(ii)-1.}} $i'$ is labelled, then $a\in B$. Otherwise, if $a \not \in B$, then in the $(u+1)$-st ``for loop'' iteration of $ALG$, $a$ will not be allocated to $i'$ at the end \david{ of} this $(u+1)$-st iteration by the process of SDMT-2. Thus $a \in B$, and since the trading graph $TG(i',C^{i'}_j\backslash\{a\},[\numagents])$ after the ``for loop'' in $ALG$ terminates is exactly the same as $TG(i',C^{i'}_j\backslash\{a\},\tau([u+1]))$ at the end of the $(u+1)$-st iteration\david{, it follows that} $i'$ will not be matched to $a$ at the end of the ``for loop'' of $ALG$, contradiction.

 \item[]\vspace*{1mm} \piotr{{\bf Case 2-(ii)-2.}}  $i'$ is unlabelled and $B=\{a\}$. Then there is a path from $i'$ to a labelled agent in $TG(i',C^{i'}_j\backslash \{a\},\tau([u]))$. Otherwise, all the agents reachable from $i'$ are frozen after the $u$-th iteration of the ``for loop''. This means that the allocations of those agents are fixed, because all the objects in their indifference class have been allocated by Observation \ref{obs:ind2}. Thus, $TG(i',C^{i'}_j\backslash\{a\},\tau([u]))$ should be the same as $TG(i',C^{i'}_j\backslash\{a\},\tau([u+1]))$. However, since $\tau(u+1)$ is reachable from $i'$ in $TG(i',C^{i'}_j\backslash\{a\},\tau([u+1]))$, while $\tau(u+1)$ does not appear in $TG(i',C^{i'}_j\backslash\{a\},\tau([u]))$, we reach a contradiction.

  \item[] \vspace*{1mm} \piotr{{\bf Case 2-(ii)-3.}}  $i'$ is unlabelled and $B=\{c\}$, where $c\neq a$. Then there is a path from $i'$ to a labelled agent in $TG(i',\{a\},\tau([u]))$. Otherwise, $a$ and the agent matched to $a$ is frozen at the end of the $u$-th ``for loop'' iteration in $ALG$. This means that $a$ will not be matched to $i'$ at the end of the ``for loop'' of $ALG$, contradiction.
\end{itemize}
\vspace*{-1mm}
 As a result, in either of the above three cases, there is a path from $i$ to a labelled agent in $TG(i,a,\tau([u]))$ at the end of the $u$-th ``for loop'' iteration in $ALG$. Namely, for case \piotr{{\bf 2-(ii)-1}}, $i$ points to $i'$, which is labelled in $TG(i,a,\tau([u]))$; for case \piotr{{\bf 2-(ii)-2}}, $i$ points to $i'$ in $TG(i,a,\tau([u]))$ and there is a path from $i'$ to a labelled agent in $TG(i',C^{i'}_j\backslash\{a\},\tau([u]))\subseteq TG(i,a,\tau([u]))$. Finally, for case \piotr{{\bf 2-(ii)-3}}, suppose $a$ is assigned to $i'''$ at the end of the $u$-th ``for loop'' iteration, then there is a path from $i'''$ to a labelled agent in $TG(i,a,\tau([u]))$ and $i$ points to $i'''$ in $TG(i,a,\tau([u]))$. This contradicts Observation \ref{obs:ind3}. Hence, $i$ must be matched \piotr{to some object}.
\end{proof}

Let $\beta^{s}_a=\beta_a((s,Y_{-i}))$, when $ALG$ denotes the execution of Algorithm \ref{alg:randomindifference} on the original instance and $Y_{-i}$ is fixed and $Y_i=s$.  Note that $\beta^{\theta}_a=w_i(1-g(\theta))/F$. \piotr{This last equality is easy to check in cases 1, 2, 3 and 5 of the definition of threshold $\theta$. In case 4, we note that $w_i(1-g(\theta))=w_{i''}(1-g(Y_{i''}))$ for some agent $i'' \not = i$. And because $\beta^{\theta}_a$ is the value of the dual variable for object $a$ when $ALG$ is run with $Y_i=\theta$, case 4 means that $\beta^{\theta}_a=w_i(1-g(\theta))/F$, despite the fact that object $a$ might not necessarily be assigned to agent $i$ (however, agent $i$ will be assigned some object).}

\piotr{We will now turn our attention to proving the monotonicity lemma.}

\begin{lemma}[Monotonicity Lemma]\label{lem:monotonicityindifference}
Given $Y_{-i}$, for all choices of $Y_i$, $\beta^{Y_i}_a\ge \beta^{\theta}_a$.
\end{lemma}

\piotr{Before presenting the full formal proof, we will first sketch the main ideas behind the proof.} The difficulty of the proof of the monotonicity lemma still lies in case 4 \piotr{from the definition of threshold $\theta$}. We prove it in three steps. Recall that $\tau$ ($\sigma$, respectively) is the order of agents under the execution of $ALG$ ($ALG_{-i}$, respectively). Let $\sigma(u)=i''$ in $ALG_{-i}$. Observe that the monotonicity lemma means that $a$ is allocated to an agent prior to $i''$ or to $i''$.  
The proof of this is easy in the case where $Y_i>\theta$. To see this, note that $ALG$ and $ALG_{-i}$ result in the same tentative allocation at the end of their $u$-th loop, since $i$ is inserted back after $i''$.
Hence, we only need to consider the case where $Y_i<\theta$, which implies that $i$ is inserted back prior to $i''$.
\begin{itemize}
\item Firstly, we prove \piotr{in Claim \ref{claim-noone-better} below,} that no agent, except $i$, is allocated a better object in $ALG$ compared to $ALG_{-i}$. The argument is by contradiction: suppose there exists an agent $i'''$ who receives a better object in $ALG$ than in $ALG_{-i}$, \piotr{then} $i$ must be inserted before $i'''$. Consequently, there exists an agent $s$ prior to $i'''$ who will get a worse object in $ALG$ than in $ALG_{-i}$. Based on this fact, and using an alternating path argument, it is proved that there exists a path from $s$ to $i'''$ in $s$'s trading graph constructed from a higher indifference class of $s$ (than $s$'s allocated indifference class in $ALG$) after  $i'''$ is allocated in $ALG$. This contradicts the fact that this path should not exist since the graph from that higher indifference class is frozen. 
\item Secondly, \piotr{we prove in Claim \ref{claim:monotoneindifference} below, that} if $i'$ gets a worse object in $ALG$ compared to $ALG_{-i}$, we prove that $a$ must be allocated to an agent prior to $i'$, which is in turn prior to $i''$. The reason is as follows: by Observation~\ref{obs:ind2}, $a$ must be allocated and frozen before $i'$ is considered in $ALG$. Then, if $i'$ gets an object in $ALG$ in the same indifference class as $a$, then we prove that there exists an agent $s^*$ prior to $i''$, and suppose $\tau(u^*)=s^*$, such that  there is a path from $s^*$ to $i'$ in $TG(s^*,C^{s^*}_{j^*},\tau([u^*]))$ at the end of the $u^*$-th ``for loop'' iteration of $ALG$. Here, $C^{s^*}_{j^*}$ is the indifference class in which $s^*$ is allocated an object in $ALG_{-i}$. As a consequence, by Observation \ref{obs:ind2}, $a$ is allocated to an agent prior to $s^*$ and all the agents reachable from $s^*$ in $TG(s^*,C^{s^*}_{j^*},\tau([u^*]))$ are frozen, then $a$ will finally be allocated to an agent prior to $s^*$ in $ALG$. This means that $a$ is allocated to an agent prior to $i''$. 
\end{itemize} This reasoning gives the monotonicity lemma, which together with dominance lemma is used to prove Lemma \ref{lem:dualfeasible1}.

\begin{proof} \piotr{(full proof of the Monotonicity Lemma, Lemma \ref{lem:monotonicityindifference})} \\
\noindent
{\bf Case 1.} \piotr{(Corresponding to case 1 in the definition of threshold $\theta$.)} If $a$ is unmatched in $ALG_{-i}$, or if $a$ is matched to $i'$ in $ALG_{-i}$ and  $i'$ is labelled, or $a$ is matched to $i'$ in $ALG_{-i}$ and  there is a path from $i'$ to a labelled agent in $TG(i',C^{i'}_j\backslash\{a\}, [\numagents])$, then $\theta=1$ and $\beta^{\theta}_a=w_i(1-g(\theta))/F=0$,  so $\beta^{Y_i}_a\ge \beta^{\theta}_a=0$.

\noindent
{\bf Case 2.} \piotr{(Corresponding to cases 2 and 3 in the definition of threshold $\theta$.)} If $a$ is matched to $i'$ in $ALG_{-i}$, there is no path from $i'$ to a labelled agent in $TG(i',C^{i'}_j\backslash\{a\}, [\numagents]\backslash\{i\})$. Suppose $\sigma(u)=i''$ in $ALG_{-i}$. \piotr{Notice that $\tau([u+1]) = \sigma([u])$. Then by Observation \ref{o:same_traging_graph}}, the trading graph  $TG(i',C^{i'}_j\backslash\{a\}, \sigma([u]))$ at the end of the $u$-th ``for loop'' iteration is the same as $TG(i',C^{i'}_j\backslash\{a\}, [\numagents]\backslash\{i\})$ at the termination of the ``for loop'' in $ALG_{-i}$. Otherwise, the $TG(i',C^{i'}_j\backslash\{a\},$ $\sigma([u]))$ is not frozen after the $u$-th ``for loop'' iteration of $ALG_{-i}$, meaning that there is a path from $i'$ to a labelled agent in  $TG(i',C^{i'}_j\backslash\{a\},$ $\sigma([u]))$. Therefore, either $i'$ will reach an agent inferior to $i''$ or a labelled agent in $TG(i',C^{i'}_j\backslash\{a\},$ $[\numagents]\backslash\{i\})$ by SDMT-2. This contradicts the definition of $i''$.

\noindent
\piotr{{\bf Case 3.} (Corresponding to case 5 in the definition of threshold $\theta$.)}
 Suppose that equation $w_i(1-g(y))=w_{i''}(1-g(Y_{i''}))$ does not have a solution, which means that $\theta=0$ and $w_i(1-g(Y_i))/F<w_{i''}(1-g(Y_{i''}))/F$, for any $Y_i\in [0,1]$. This shows that the process is the same for agents prior to agent $i''$ until the end of the $u$-th ``for loop'' iteration in $ALG$ and $ALG_{-i}$. Since there is no path from $i'$ to a labelled agent in $TG(i',C^{i'}_j\backslash\{a\}, \sigma([u]))$, the agents reachable from $i'$ are frozen. Hence, $a$ will be finally still allocated to $i'$ in $ALG$, implying $\beta^{Y_i}_a=w_{i'}(1-g(Y_{i'}))/F\ge w_{i''}(1-g(Y_{i''}))/F>\beta^{\theta}_a=w_i(1-g(0))/F$.

\noindent
{\bf Case 4.} \piotr{(Corresponding to case 4 in the definition of threshold $\theta$.)}
 Now consider the last case that equation $w_i(1-g(y))=w_{i''}(1-g(Y_{i''}))$ has a solution, then  $\beta^{\theta}_a=w_i(1-g(\theta))/F=w_{i''}(1-g(Y_{i''}))/F$. Consider the following three cases:
 
\noindent
\piotr{{\bf Case (4-i):}} If $Y_i> \theta$, this means $w_i(1-g(Y_i))/F< w_{i''}(1-g(Y_{i''}))/F$, and the analysis of this case is the same as above (the case that equation $w_i(1-g(y))=w_{i''}(1-g(Y_{i''}))$ does not have a solution), since $i$ will select objects after $i''$. Thus, we have $\beta^{Y_i}_a=w_{i'}(1-g(Y_{i'}))/F\ge \beta^{\theta}_a$.

\noindent
\piotr{{\bf Case (4-ii):}} If $Y_i<\theta$, then  $w_i(1-g(Y_i))/F> w_{i''}(1-g(Y_{i''}))/F$, which means that  $i$ is prior to $i''$ in $ALG$. We have the following claim:
\begin{claim}\label{claim-noone-better}
No agent can get a better object in $ALG$ than in $ALG_{-i}$ after inserting $i$ into some position from $1$ to $u$.
\end{claim}
\begin{proof}
 Suppose, towards a contradiction, that there exists an agent getting a better object, and let $k$ be the smallest position where such agents are placed in $ALG$. Then $i$ must be inserted in a position before $k$ (otherwise, the process is the same for the first $k$ agents in $ALG_{-i}$ and $ALG$, so agent $\tau(k)$ can not get a better object). Let $i'''=\tau(k)$ and suppose that $i'''$ gets an object $b$ in $ALG$ and object $c$ in $ALG_{-i}$, where $b\succ_{i'''}c$. Observe that $\sigma(k-1)=i'''$ in $ALG_{-i}$. Suppose that $b\in C^{i'''}_j$ and consider the trading graph $TG(i''',C^{i'''}_j,\sigma([k-1]))$ at the end of the $(k-1)$-st ``for loop'' iteration of $ALG_{-i}$.
  
Let $S$ be the set of agents reachable from $i'''$ in $TG(i''',C^{i'''}_j,\sigma([k-1]))$ at the end of the $(k-1)$-st ``for loop'' iteration of $ALG_{-i}$. Note that any agent in $S$ is prior to $i'''$. Any agent in $S$ is allocated only one object and frozen in $ALG_{-i}$. Since in $ALG$, $b$ is allocated to $i'''$, then in the $k$-th ``for loop'' iteration of $ALG$, $i'''$ will be allocated some objects in $C^{i'''}_j$. This means that some agent in $S$ will get worse object compared to the allocation in $ALG_{-i}$. 

The reason is as follows: \piotr{no agent} can get \piotr{a} better \piotr{object} by the definition of $k$. If all the agents in $S$ can remain the same in $ALG$ compared with $ALG_{-i}$ (i.e., get the objects in the same indifference class in $ALG$ and in $ALG_{-i}$), then 
the only possible allocation of $S$ in $ALG$ is reallocating all the objects matched to $S$ in $ALG_{-i}$ to $S$ again such that each agent gets exactly one object. If there is some extra object $e$ in $ALG$ allocated to an agent from $S$ in $ALG$, then $e$ must be allocated to some agent $j$ in $ALG_{-i}$. Since $e$ in $ALG$ is allocated to some agent in $S$, thus, $j$ can be reached by some agent in $S$ in $ALG_{-i}$. Thus, $j\in S$, which leads to a contradiction.
All the objects in $C^{i'''}_j$ have been allocated to some agents in $S$. In $ALG$, we will need to allocate $|S|$ objects to $S\cup\{i'''\}$ agents because some objects owned by $S$ in $ALG_{-i}$ will be allocated to agent $i'''$. This is not possible, which gives a contradiction.
  
Let $s$ be an agent in $S$ who gets a worse object and there is a path from $s$ to $i'''$ in the trading graph $TG(s,d,\tau[k])$ at the end of the $k$-th ``for loop'' iteration in $ALG$, where $d$ is the allocated object of $s$ in $ALG_{-i}$. (Such an agent must exist: it can be found by the following procedure. Suppose $d_1\simeq_{s_1} b$ owned by $s_1$ in $ALG_{-i}$ is allocated to $i'''$ in $ALG$ at the end of the $k$-th ``for loop'' iteration in $ALG$. If $s_1$ gets worse in $ALG$ compared to $ALG_{-i}$, then $s_1$ is \piotr{the} agent we are looking for. Otherwise, $s_1$ will be allocated object $d_2$ owned by $s_2 \in S$ in $ALG_{-i}$ at the end of the
$k$-th ``for loop'' iteration of $ALG$. If $s_2$ gets a worse object, then $s_2$ is \piotr{the} agent we are looking for. Otherwise, we continue with this procedure. By finiteness of the set $S$ and by the fact that \piotr{the} agents in $S$ own $|S|$ objects in $ALG_{-i}$, these objects will be allocated to agents in $S\cup \{i'''\}$ in $ALG$, and one of these objects will be allocated to $i'''$. Thus, we can find such an agent. The path from $s$ to $i'''$ in the trading graph $TG(s,d,\tau[k])$ at the end of the $k$-th ``for loop'' iteration in $ALG$ is just the reverse path by the above procedure). Suppose $\tau(\ell)=s$ in $ALG$ and $d\in C^s_{h}$. Consider the $\ell$-th ``for loop'' iteration in $ALG$: all the agents reachable from $s$ in $TG(s,C^s_h,\sigma([\ell]))$ are frozen and prior to agent $s$ since $s$ does not obtain any object in the indifference class $C^s_h$.
   
 This contradicts the fact that there is a path from $s$ to $i'''$ (which is inferior to $s$) in the trading graph $TG(s,d,\tau[k])$ at the end of the $k$-th ``for loop'' iteration in $ALG$.
 \end{proof}
 \begin{claim}\label{claim:monotoneindifference}
 Object $a$ must be allocated to an agent prior to $i''$ or to $i''$, \piotr{that is, we must have $\beta^{Y_i}_a\ge w_{i''}(1-g(Y_{i''}))=\beta^{\theta}_a$.}
\end{claim}
\begin{proof}
 Suppose $\sigma(u_1)=i'$ and $\sigma(u)=i''$ in $ALG_{-i}$.
The following cases are considered:

{\bf Case (1).} If $i'$ gets worse, meaning he gets a worse object in $ALG$ than $a$ in $ALG_{-i}$, then $\tau(u_1+1)=i'$ in $ALG$ ($i$ is inserted back prior to $i'$). Thus, all the agents reachable from $i'$ in $TG(i',a,\tau([u_1+1]))$ are frozen and the agent who owns $a$ will finally get $a$. This agent is prior to $i'$, giving that $\beta^{Y_i}_a\ge w_{i'}(1-g(Y_{i'}))\ge w_{i''}(1-g(Y_{i''}))=\beta^{\theta}_a$.

{\bf Case (2).} If $i'$ gets $a$ in $ALG$, then we are done. Otherwise, suppose $i'$ gets an object $a'\simeq_{i'} a$, $a', a \in C^{i'}_j$ in $ALG$. Denote by $S^*$ the set of agents reachable from $i'$ in $TG(i',C^{i'}_j\backslash\{a\},\sigma([\numagents-1]))$ at the end of the ``for loop'' of $ALG_{-i}$ (note that $\sigma([\numagents-1])=[\numagents]\backslash\{i\}$). If no one in $S^*$ gets worse in $ALG$ than in $ALG_{-i}$, then $a$ must be allocated to some agent in $S^*$. The reason is similar to the above argument.  All agents in $S^*$ get exactly one object. If $a$ is not allocated in $S^*$, no one gets worse in $S^*$, and there must be an extra object $b$ allocated to some agent $j$ in $S^*$. No matter whom $b$ is allocated to in $ALG_{-i}$, there is a path from $j$ to this agent. Hence, this agent belongs to $S^*$, a contradiction.
Note that, by the definition of $i''$,  for any $s\in S^*$, $\sigma^{-1}(s)>\sigma^{-1}(i'')$ ($\sigma^{-1}(s)$ denotes the order of $s$ in $\sigma$ or in $ALG_{-i}$) implies that $w_{s}(1-g(Y_{s}))\ge w_{i''}(1-g(Y_{i''}))$.
 Therefore $\beta^{Y_i}_a\ge w_{i''}(1-g(Y_{i''}))=\beta^{\theta}_a$.  
 
Otherwise, by the previous argument, there exists $s^*\in S^*$ who gets worse in $ALG$ compared to $ALG_{-i}$, and there is a path from $s^*\in S^*$ to $i'$ in $TG(s^*,d^*,[\numagents])$  at the end of $ALG$, where $d^*$ is the allocation of $s^*$ in $ALG_{-i}$. If $s^*$ is prior to $i'$, then by similar argument as above, there should be no path from $s^*$ to an agent inferior to $s^*$ (constructed from the objects in $L(s^*)$ no worse than $d^*$) at the end of $ALG$, a contradiction. Hence, $s^*$ can only be inferior to $i'$. Suppose $\tau(u^*)=s^*$ and $d^*\in C^{s^*}_{j^*}$, then we know that there is a path from $s^*$ to $i'$ in $TG(s^*,C^{s^*}_{j^*},[\numagents])$ at the end of the ``for loop'' of $ALG$. Next we will prove the following statement (which we denote as $(*)$): $$\text{There is a path from}\, s^*\, \text{to}\, i'\, \text{in}\, TG(s^*,C^{s^*}_{j^*},\tau([u^*]))\, \text{at the end of}\, u^*\text{-th ``for loop'' of}\, ALG.$$
 If \piotr{$(*)$} is true, then all the agents reachable from $s^*$ in $TG(s^*,C^{s^*}_{j^*},\tau([u^*]))$ are frozen, and $a$ is allocated to an agent prior to $s^*$ due to Observation \ref{obs:ind2}. Then, we have that   $\beta^{Y_i}_a\ge w_{s^*}(1-g(Y_{s^*}))\ge w_{i''}(1-g(Y_{i''}))=\beta^{\theta}_a$ since $s^*\in S^*$. 
 
 \noindent
 Suppose \piotr{finally that $(*)$} is not true, then  all the agents $U^*$ \piotr{that are} reachable from $s^*$ in $TG(s^*,d^*,\tau([u^*]))$ have been frozen. $U^*$ will remain the same until the end of $ALG$ and $i'\notin U^*$. However, by the definition of $S^*$ and by $s^* \in S^*$, there is a path from $s^*$ to $i'$ in $TG(s^*,d^*,[\numagents])$  at the end of  $ALG$, meaning $i'\in U^*$, a contradiction.
\end{proof}
\noindent
By Claim \ref{claim:monotoneindifference}, we know that if $Y_i<\theta$ then $\beta^{Y_i}_a\ge w_{i''}(1-g(Y_{i''}))/F= \beta^{\theta}_a$.

\noindent
\piotr{{\bf Case (4-iii):}} If $Y_i= \theta$, this means that $w_i(1-g(Y_i))/F= w_{i''}(1-g(Y_{i''}))/F$. If $i''<i$, the case is same as if $Y_i>\theta$. Otherwise, it falls into the case $Y_i<\theta$.

To summarise, for all choices of $Y_i$, $\beta^{Y_i}_a\ge \beta^{\theta}_a$.
\end{proof}

\piotr{Let us recall that $F=(1-\frac{1}{e})$ and $g(y) = e^{y-1}$. Observe that $\int g(y)dy = g(y) + C$, where $C$ is any fixed constant. 
Then it is not difficult to see that
\begin{equation}\label{equ:gF1}
\mbox{for each } t \in [0,1] \,:\,\, \int_0^{t}g(y)dy
+1-g(t) = F
\end{equation}
}

\begin{lemma}[\cite{Devanur-etal}]\label{lem:dualfeasible1}
$\forall (i,a)\in E$, $\E_{\vec{Y}}(\alpha_i(\vec{Y})+\beta_a(\vec{Y}))\ge w_i$.
\end{lemma}

\begin{proof}
For fixed choices of $Y_{-i}$, by the Dominance Lemma (Lemma \ref{lem:dominanceindifference}), $i$ is matched whenever $Y_i< \theta$. Hence,
$$\E_{Y_i}(\alpha_i(\vec{Y}))\ge w_i\int_0^{\theta}g(y)dy/F.$$
By the Monotonicity Lemma (Lemma \ref{lem:monotonicityindifference}), $\beta_a(\vec{Y})=\beta^{Y_i}_a\ge\beta^{\theta}_a=w_i(1-g(\theta))/F$, for any $Y_i\in [0,1]$, then
$$\E_{Y_i}(\beta_a(\vec{Y}))\ge w_i(1-g(\theta))/F.$$
Therefore, note that  by formula (\ref{equ:gF1}), we have
$$\E_{Y_i}(\alpha_i(\vec{Y})+\beta_a(\vec{Y}))\ge w_i\int_0^{\theta}g(y)dy/F+w_i(1-g(\theta))/F = w_i.$$
As a result, $\E_{\vec{Y}}(\alpha_i(\vec{Y})+\beta_a(\vec{Y}))\ge w_i$.
\end{proof}

From Lemma \ref{lem:factor1} and Lemma \ref{lem:dualfeasible1}, we have the following theorem.
\begin{theorem}\label{thm:weightindifference}
Algorithm \ref{alg:randomindifference}  achieves an approximation ratio of $\frac{e}{e-1}$ for weighted agents with indifferences.
\end{theorem}

\section{Online interpretation and lower bounds} \label{sec:bounds}


\piotr{We will first provide here an ``online'' flavour interpretation of the weighted version of our problem. We interpret it in the following way. An administrator holds all the objects, and all agents with unknown preference lists are applicants for these objects. We assume that weights are private information of each agent, but that they cannot overstate their weights, a so-called no-overbidding assumption. Applicants are interviewed one-by-one in a random order.
A decision about each particular applicant is to be made immediately after the interview. 
During the interview, the applicant selects his \david{favourite} object among the available remaining objects if there exists one in his preference list and must be allocated (matched to) that object because we consider only truthful mechanisms.\footnote{We can extend this setting to the case where the administrator can decide whether to let the applicant select his \david{favourite} object or to reject this applicant, meaning that the applicant gets nothing. In this more general problem, it is not difficult to prove that for any fixed order of the applicants, the decision that the administrator does not reject any applicant will maximise the number of matched applicants. Therefore, this more general problem is reduced to the setting where the administrator lets each applicant select his \david{favourite} object, and \david{hence} our lower bound from Section \ref{sec:bounds} also applies to this setting.} %
This applicant will not be interviewed again. The administrator can know the number of matched applicants interviewed so far, but is unaware whether yet unseen applicants will be matched or not.
Our goal is to find the optimal strategy, that is a (random) arrival order of agents that maximises the ratio between the total weight of matched agents and the maximum weight of a matching if all the agents’ preference lists are known in advance.} \\

\piotr{We will now describe the required preliminaries that will be used in the remainder of this section to prove the lower bounds.} \\

\noindent
{\bf Preliminaries.} We will use Yao's minmax principle, \piotr{see \cite[Proposition 2.5 (page 35)]{MOtwaniR95} and} \cite{Ya77}, to obtain a non-trivial lower bound for universally truthful and Pareto optimal mechanisms and another lower bound for an ``online'' version of our problem. We first need some preliminaries.

Let us fix the number of agents $\numagents$ and the number of objects $\numobjects$.  The number of distinct instances and the number of deterministic truthful and Pareto optimal mechanisms are finite. Denote by $\mathcal{T}$ the set of deterministic truthful and Pareto optimal mechanisms with input size $\numagents$ and $\numobjects$, and $\mathcal {I}$ the set of instances with input size $\numagents$ and $\numobjects$. Let $\mathrm{P}$ and $\mathrm{Q}$ denote the set of  probability distributions on $\mathcal{T}$ and $\mathcal{I}$, respectively. Denote $\mathbb{E}_{p,q}(r(T_p,I_q))$ as the inverse of approximation ratio when 
\bahar{the input $I_q$ is sampled according to the distribution $q\in \mathrm{Q}$ and a universally truthful and Pareto optimal mechanism $T_p$ is sampled according to the distribution $p\in \mathrm{P}$.} Then the minmax theorem \cite{Ya77} states the following:
$$\min_{q\in\mathrm{Q}}\max_{p\in\mathrm{P}}\mathbb{E}_{p,q}(r(T_p,I_q))=\max_{p\in\mathrm{P}}\min_{q\in\mathrm{Q}}\mathbb{E}_{p,q}(r(T_p,I_q))$$
and
$$\min_{q\in\mathrm{Q}}\max_{T\in\mathcal{T}}\mathbb{E}_{q}(r(T,I_q))=\max_{p\in\mathrm{P}}\min_{I\in\mathcal{I}}\mathbb{E}_{p}(r(T_p,I)).$$
As a consequence, for any $q\in\mathrm{Q}$ and $p\in\mathrm{P}$, we have
$$\max_{T\in\mathcal{T}}\mathbb{E}_{q}(r(T,I_q))\ge\min_{I\in\mathcal{I}}\mathbb{E}_{p}(r(T_p,I)).$$
This inequality states that an upper bound on the inverse of the approximation ratio of the best universally truthful and Pareto optimal mechanism $T_p$ on the worst instance is upper bounded by the inverse of the approximation ratio of the best deterministic truthful and Pareto optimal mechanism on 
\bahar{a randomly chosen instance.} 
Hence, in order to bound $\min_{I\in\mathcal{I}}\mathbb{E}_{p}(r(T_p,I))$, we only need to construct an appropriate random instance and 
\bahar{compute}
the upper bound of the best deterministic truthful and Pareto optimal mechanism on this random instance.
Consider the {\em triangle instance} where $\agentset=\{1,2,\cdots,\numagents\}$ and $\objectset=\{\object_1,\object_2,\cdots,\object_{\numagents}\}$, and an agent $i$'s preference ordering is $\object_1\succ_i \object_2\succ_i\cdots\succ_i \object_i$, for any $i \in \agentset$.

Let $\mathcal{S}$ 
\bahar{denote}
the set of all the permutations of agents' preference lists of the triangle instance. Consider now a random instance $S_{uni}$ as the uniform distribution of $\mathcal{S}$. It is obvious that the output of any serial dictatorship mechanism (which is \piotr{a} deterministic, truthful and Pareto optimal \piotr{mechanism, defined by a specific fixed order of the agents}) running on $\mathcal{S}$ is the same. Hence, for any serial dictatorship mechanism (SDM),  $\mathbb{E}_{uni}(r(SDM,S_{uni}))$ is equal to the inverse of the approximation ratio of RSDM, which is just SDM with 
\bahar{the order of agents chosen uniformly at random,}
when running on the triangle instance. \\

\noindent
{\bf Online lower bound.} We now apply these preliminaries to the online version of our problem. Recall that applicants in this online problem are truthful due to the truthfulness of serial dictatorship mechanism. The strategy of the administrator is a random order in which the applicants are interviewed. More precisely, let $\Pi$ denote the set of all the permutations of applicants and \piotr{$P(\Pi)$} be the set of probability distributions on $\Pi$. Let $\Pi_p$ be a random order of applicants, where the order is selected according to the distribution $p\in \piotr{P(\Pi)}$ on $\Pi$, and then the strategy set of the administrator is $\{\Pi_p \piotr{\,\, : \,\,} p\in \piotr{P(\Pi)}\}$.
We will show that the best strategy for the administrator is to select applicants' order uniformly at random.
\begin{theorem}\label{thm:OSDM}
 The best strategy for the administrator in the online problem is to select the applicants' order uniformly at random.
Thus, any other randomised strategy, than the one used in Algorithm \ref{alg:randomindifference}, would lead to an approximation guarantee worse than $\frac{e}{e-1}$.
\end{theorem}

\begin{proof} This proof is similar to the classical proof from \cite{kvv}. In particular it uses the same class of instances.
Let $\mathbb{E}_{p,q}(r(\Pi_p,I_q))$ be the inverse of the approximation ratio when the random order is $\Pi_p$ and the random instance is $I_q$, and let $\Pi_{uni}$ denote the uniform order. By the approximation ratio of RSDM, for any $\inputsetting$, $\mathbb{E}_{uni}(r(\Pi_{uni},I))\ge \frac{e-1}{e}$. Now for upper bound of $\mathbb{E}_{p,q}(r(\Pi_p,I_q))$, by Yao's principle \piotr{\cite[Proposition 2.5]{MOtwaniR95},} $\max_{T\in\Pi}\mathbb{E}_{q}(r(T,I_q))\ge\min_{I\in\mathcal{I}}\mathbb{E}_{p}(r(T_p,I)).$
Recall that $S_{uni}$ is the uniform distribution over $\mathcal{S}$. Then we need to upper bound $\max_{T\in\Pi}\mathbb{E}_{q}(r(T,S_{uni}))$, which in fact is equal to the inverse of the approximation ratio obtained by running RSDM on the triangle instance, which is $\frac{e-1}{e}$. The argument is as follows.
  Suppose object $\object_k$ is allocated by RSDM with probability $p_k\le 1$ on the triangle instance. Then, because there are $\numagents -k +1$ agents with $\object_k$ in their preference lists, each such agent obtains $\object_k$ with equal probability $\frac{p_k}{\numagents-k+1}$. Therefore, agent $i$ is allocated an object with probability $\sum_{j=1}^i\frac{p_j}{\numagents-j+1}$, which is at most $\min\{1,\sum_{j=1}^i\frac{1}{\numagents-j+1}\}$. Now, summing over all the agents, by a simple calculation we get that the expected cardinality of \bahar{the number of} allocated agents is at most $\numagents(1-\frac{1}{e})$, for large enough $\numagents$. Hence, the approximation ratio is tight.
\end{proof}

\noindent
{\bf Lower bound for randomised mechanisms.} If we can prove that the output of any deterministic truthful and Pareto optimal mechanism running on $\mathcal{S}$ is the same as that of
\bahar{SDM}
then $\max_{T\in\mathcal{T}}\mathbb{E}_{q}(r(T,I))=1-\frac{1}{e}$. To show our lower bound it suffices to show that the sum of the sizes of all the matchings 
\bahar{returned}
by any deterministic truthful and Pareto optimal mechanism 
\bahar{executed}
on $\mathcal{S}$ is smaller than that of \bahar{returned by} any \bahar{SDM}
\bahar{executed}
on $\mathcal{S}$. Then $\max_{T\in\mathcal{T}}\mathbb{E}_{q}(r(T,I))=1-\frac{1}{e}$.  We use $\#^{\phi}(\mathcal{S})$ to denote the sum of \bahar{the} sizes of all the matchings 
\bahar{returned}
by mechanism $\phi$ when 
\bahar{executed}
on $\mathcal{S}$. 
We want to
to prove that $\#^{\phi}(\mathcal{S})\le \#^{SDM}(\mathcal{S})$, for any $\numagents$ and $\numobjects$ and for any universally truthful and Pareto optimal mechanism $\phi$.
We can prove this inequality assuming $\numagents=\numobjects=3$, which gives \bahar{us} the lower bound of $\frac{18}{13}$ for any universally truthful and Pareto optimal mechanism.
\begin{theorem}\label{lowerbound:n=3}
For any deterministic truthful and Pareto optimal mechanism $\phi$, $\#^{\phi}(\mathcal{S})\le 13$, when $\numagents=3$. Thus, any universally truthful and Pareto optimal mechanism for this problem has an approximation ratio of at least $\frac{18}{13}$.
\end{theorem}
\begin{proof}
Suppose the agents are $1$, $2$, $3$ and objects are $a$, $b$, $c$. We use the notation
$\left(\begin{array}{ccc}
\underline{a}& b & c\\
a & \underline{b} & \\
a& b& \underline{c}
\end{array}\right)$ to denote assignments that allocate $a$ to agent $1$, $b$ to agent $2$ and $c$ to agent $3$, where row $i$ denotes agent $i$'s preference list and preference ordering is the increasing order of column indices, $i=1,2,3$. If there are no underlines of the objects, then  this notation denotes the input of mechanism. Note that in this setting, $\mathcal{S}=\left\{\left(\begin{array}{ccc}
a&  & \\
a & b & \\
a& b& c
\end{array}\right)\right.$, $\left(\begin{array}{ccc}
a&  & \\
a & b &c \\
a& b&
\end{array}\right)$, $\left(\begin{array}{ccc}
a& b & \\
a &  & \\
a& b&c
\end{array}\right)$, $\left(\begin{array}{ccc}
a& b & \\
a & b &c \\
a& &
\end{array}\right)$, $\left(\begin{array}{ccc}
a& b & c\\
a &  & \\
a& b&
\end{array}\right)$, $\left. \left(\begin{array}{ccc}
a&  b&c \\
a & b & \\
a& &
\end{array}\right)\right\}$. We would like to show that for any deterministic truthful and Pareto optimal mechanism $\phi$, $\#^{\phi}(\mathcal{S})\le 13$. \piotr{Without loss of generality,} suppose $\left(\begin{array}{ccc}
\underline{a}& b & c\\
a & \underline{b} & c\\
a& b& \underline{c}
\end{array}\right)$, and we will consider the following two cases:

\noindent
{\bf Case (i):} If $\left(\begin{array}{ccc}
a& b & \underline{c}\\
a & \underline{b} & \\
\underline{a}& b& c
\end{array}\right)$, then we will show that
$\left(\begin{array}{ccc}
\underline{a}& b & \\
a & \underline{b} & c\\
a& b& \underline{c}
\end{array}\right)$. (Observe that the first agent must get $a$ because otherwise we have contradiction with truthfulness by $\left(\begin{array}{ccc}
\underline{a}& b & c\\
a & \underline{b} & c\\
a& b& \underline{c}
\end{array}\right)$.)

Now, if $\left(\begin{array}{ccc}
\underline{a}& b & \\
a & \underline{b} & c\\
a& b& \underline{c}
\end{array}\right)$ would not hold then $\left(\begin{array}{ccc}
\underline{a}& b & \\
a &b & \underline{c}\\
a& \underline{b}&c
\end{array}\right)$.   Then we obtain $\left(\begin{array}{ccc}
a& b & \\
a & b & \\
\underline{a}& b& c
\end{array}\right)$. The reason is as follows:   $\left(\begin{array}{ccc}
a& b & \underline{c}\\
a & \underline{b} & \\
\underline{a}& b& c
\end{array}\right)$ implies that the first agent in the input $\left(\begin{array}{ccc}
a& b & \\
a & b & \\
a & b& c
\end{array}\right)$ cannot get any object by truthfulness.  Similarly, from  $\left(\begin{array}{ccc}
\underline{a}& b & \\
a &b & \underline{c}\\
a& \underline{b}&c
\end{array}\right)$, the second agent in $\left(\begin{array}{ccc}
a& b & \\
a & b & \\
a & b& c
\end{array}\right)$ cannot get any object by truthfulness. Thus we have that $\left(\begin{array}{ccc}
a& b & \\
a & b & \\
\underline{a}& b& c
\end{array}\right)$, which is a contradiction to Pareto optimality.

By a similar argument, we have
$\left(\begin{array}{ccc}
a& b & \\
a & \underline{b} & \\
\underline{a}& b& c
\end{array}\right)$,
$\left(\begin{array}{ccc}
\underline{a}& & \\
a & \underline{b} & c\\
a& b& \underline{c}
\end{array}\right)$  and
$\left(\begin{array}{ccc}
a& &\\
a & \underline{b} & \\
\underline{a}& b& c
\end{array}\right)$.
From $\left(\begin{array}{ccc}
a& b & \\
a & \underline{b} & \\
\underline{a}& b& c
\end{array}\right)$,
we know the size of the matching output from $\left(\begin{array}{ccc}
a& b & \\
a &  & \\
a& b& c
\end{array}\right)$ is $2$.
From $\left(\begin{array}{ccc}
a& b &\underline{c} \\
a & \underline{b} & \\
\underline{a}& b& c
\end{array}\right)$,
we know the size of the matching output from $\left(\begin{array}{ccc}
a& b & c\\
a &  & \\
a&b &
\end{array}\right)$ is  $2$. From $\left(\begin{array}{ccc}
\underline{a}& b & \\
a & \underline{b} & c\\
a& b& \underline{c}
\end{array}\right)$,
 we know the size of the matching output from $\left(\begin{array}{ccc}
a& b & \\
a & b & c\\
a&
\end{array}\right)$ is $2$.
From $\left(\begin{array}{ccc}
\underline{a}& & \\
a & \underline{b} & c\\
a& b& \underline{c}
\end{array}\right)$,
we know the size of the matching output from $\left(\begin{array}{ccc}
a&  & \\
a & b & c\\
a&b
\end{array}\right)$ is at most $2$.
Thus, if the current mechanism is $\phi^1$ then $\#^{\phi^1}(\mathcal{S})\le13$.

\noindent
{\bf Case (ii):} If $\left(\begin{array}{ccc}
\underline{a}& b & c\\
a & \underline{b} & \\
a& b& \underline{c}
\end{array}\right)$, we consider the following two cases:

\noindent
{\bf Case (ii-a):} If  $\left(\begin{array}{ccc}
\underline{a}& b &\\
a & b & \underline{c}\\
a& \underline{b}& c
\end{array}\right)$,
 then $\left(\begin{array}{ccc}
\underline{a}& b &\\
a & b & \\
a& \underline{b}& c
\end{array}\right)$,
and we conclude that  $\left(\begin{array}{ccc}
\underline{a}& b &c\\
a & \underline{b} & c\\
a& b&
\end{array}\right)$,
otherwise suppose  $\left(\begin{array}{ccc}
a& \underline{b} &c\\
\underline{a} &b & c\\
a& b&
\end{array}\right)$
(since $\left(\begin{array}{ccc}
\underline{a}& b &c\\
a &\underline{b} & c\\
a& b&\underline{c}
\end{array}\right)$),
then $b$ is allocated to agent $1$ in  $\left(\begin{array}{ccc}
a& b &\\
a &b & c\\
a& b&
\end{array}\right)$.
 From $\left(\begin{array}{ccc}
\underline{a}& b &\\
a & b & \underline{c}\\
a& \underline{b}& c
\end{array}\right)$,
we know $b$ is allocated to agent $3$ in $\left(\begin{array}{ccc}
a& b &\\
a &b & c\\
a& b&
\end{array}\right)$,
a contradiction. Hence,  $\left(\begin{array}{ccc}
\underline{a}& b &c\\
a & \underline{b} & c\\
a& b&
\end{array}\right)$,
 then we know  $\left(\begin{array}{ccc}
\underline{a}& b &c\\
a & \underline{b} & \\
a& b&
\end{array}\right)$
(from  $\left(\begin{array}{ccc}
\underline{a}& b &c\\
a & \underline{b} & c\\
a& b&
\end{array}\right)$
 and  $\left(\begin{array}{ccc}
\underline{a}& b &c\\
a & \underline{b} & \\
a& b&\underline{c}
\end{array}\right)$ ),
and  $\left(\begin{array}{ccc}
\underline{a}& b &\\
a & b & \underline{c}\\
a& \underline{b}&
\end{array}\right)$ (From $\left(\begin{array}{ccc}
\underline{a}& b &c\\
a & \underline{b} & c\\
a& b&
\end{array}\right)$
and $\left(\begin{array}{ccc}
\underline{a}& b &\\
a & b & \underline{c}\\
a& \underline{b}& c
\end{array}\right)$).
Now the matching size of assignment of$\left(\begin{array}{ccc}
a& b &c\\
a &  &\\
a& b&
\end{array}\right)$
and $\left(\begin{array}{ccc}
a& b &c\\
a & b &\\
a& &
\end{array}\right)$
is both $2$ (from $\left(\begin{array}{ccc}
\underline{a}& b &c\\
a & \underline{b} & \\
a& b&
\end{array}\right)$).
 The matching size of assignment of $\left(\begin{array}{ccc}
a& b &\\
a & b &c\\
a& &
\end{array}\right)$
 is $2$ since $\left(\begin{array}{ccc}
\underline{a}& b &\\
a & b & \underline{c}\\
a& \underline{b}&
\end{array}\right)$.
 The matching size of assignment of$\left(\begin{array}{ccc}
a& b &\\
a &  &\\
a& b&c
\end{array}\right)$
 is $2$ following from $\left(\begin{array}{ccc}
\underline{a}& b &\\
a & b & \\
a& \underline{b}& c
\end{array}\right)$.
Consider the assignment of $\left(\begin{array}{ccc}
a&  &\\
a & b &c\\
a& b&c
\end{array}\right)$, no matter what the assignment is, at most one matching size of assignment of $\left(\begin{array}{ccc}
a&  &\\
a & b &\\
a& b&c
\end{array}\right)$ and $\left(\begin{array}{ccc}
a&  &\\
a & b &c\\
a& b&
\end{array}\right)$ is $3$. Denote the mechanism in this case by $\phi_2$, then $\#^{\phi_2}(\mathcal{S})\le 13$.

\noindent
{\bf Case (ii-b):} If $\left(\begin{array}{ccc}
\underline{a}& b &\\
a & \underline{b} &c\\
a& b&\underline{c}
\end{array}\right)$, recall that $\left(\begin{array}{ccc}
\underline{a}& b &c\\
a & \underline{b} &\\
a& b&\underline{c}
\end{array}\right)$ and $\left(\begin{array}{ccc}
\underline{a}& b &c\\
a & \underline{b} &c\\
a& b&\underline{c}
\end{array}\right)$, consider the following two cases:

\noindent
{\bf Case (ii-b-1):} If $\left(\begin{array}{ccc}
\underline{a}& b &c\\
a & \underline{b} &c\\
a& b&
\end{array}\right)$, then $\left(\begin{array}{ccc}
\underline{a}& b &c\\
a & \underline{b} &\\
a& b&
\end{array}\right)$ since $\left(\begin{array}{ccc}
\underline{a}& b &c\\
a & \underline{b} &\\
a& b&\underline{c}
\end{array}\right)$. We know the matching sizes of assignment of $\left(\begin{array}{ccc}
a& b &c\\
a & b &\\
a& &
\end{array}\right)$ and $\left(\begin{array}{ccc}
a& b &c\\
a & &\\
a& b&
\end{array}\right)$ are both $2$. Since $\left(\begin{array}{ccc}
\underline{a}& b &\\
a & \underline{b} &\\
a& b&\underline{c}
\end{array}\right)$ due to $\left(\begin{array}{ccc}
\underline{a}& b &c\\
a & \underline{b} &\\
a& b&\underline{c}
\end{array}\right)$ and $\left(\begin{array}{ccc}
\underline{a}& b &c\\
a & \underline{b} &c\\
a& b&
\end{array}\right)$,  the matching size of assignment of $\left(\begin{array}{ccc}
a& b &\\
a &  &\\
a& b&c
\end{array}\right)$ is $2$. Since $\left(\begin{array}{ccc}
\underline{a}& b &\\
a & \underline{b} &c\\
a& b&
\end{array}\right)$ due to $\left(\begin{array}{ccc}
\underline{a}& b &c\\
a & \underline{b} &c\\
a& b&
\end{array}\right)$ and $\left(\begin{array}{ccc}
\underline{a}& b &\\
a & \underline{b} &c\\
a& b&\underline{c}
\end{array}\right)$, then  the matching size of assignment of $\left(\begin{array}{ccc}
a& b &\\
a & b &c\\
a& &
\end{array}\right)$ is $2$. Similar as the above argument, consider the assignment of $\left(\begin{array}{ccc}
a&  &\\
a & b &c\\
a& b&c
\end{array}\right)$, no matter what the assignment is, at most one matching size of assignment of $\left(\begin{array}{ccc}
a&  &\\
a & b &\\
a& b&c
\end{array}\right)$ and $\left(\begin{array}{ccc}
a&  &\\
a & b &c\\
a& b&
\end{array}\right)$ is $3$. Denote the mechanism in this case by $\phi_3$, then $\#^{\phi_3}(\mathcal{S})\le 13$.

\noindent
{\bf Case (ii-b-2):} If $\left(\begin{array}{ccc}
a& \underline{b} &c\\
\underline{a} & b &c\\
a& b&
\end{array}\right)$, recall that we have $\left(\begin{array}{ccc}
\underline{a}& b &\\
a & \underline{b} &c\\
a& b&\underline{c}
\end{array}\right)$, $\left(\begin{array}{ccc}
\underline{a}& b &c\\
a & \underline{b} &\\
a& b&\underline{c}
\end{array}\right)$ and $\left(\begin{array}{ccc}
\underline{a}& b &c\\
a & \underline{b} &c\\
a& b&\underline{c}
\end{array}\right)$. From $\left(\begin{array}{ccc}
a& \underline{b} &c\\
\underline{a} & b &c\\
a& b&
\end{array}\right)$ and $\left(\begin{array}{ccc}
\underline{a}& b &\\
a & \underline{b} &c\\
a& b&\underline{c}
\end{array}\right)$, we get $\left(\begin{array}{ccc}
a& \underline{b} &\\
\underline{a} & b &c\\
a& b&
\end{array}\right)$, then the matching sizes of assignment of $\left(\begin{array}{ccc}
a&  &\\
a & b &c\\
a& b&
\end{array}\right)$ and $\left(\begin{array}{ccc}
a& b &\\
a & b &c\\
a& &
\end{array}\right)$ are both $2$. From $\left(\begin{array}{ccc}
\underline{a}&b &c\\
a & \underline{b} &\\
a& b&\underline{c}
\end{array}\right)$ and $\left(\begin{array}{ccc}
\underline{a}& b &\\
a & \underline{b} &c\\
a& b&\underline{c}
\end{array}\right)$, we get $\left(\begin{array}{ccc}
\underline{a}& b &\\
a & \underline{b} &\\
a& b&\underline{c}
\end{array}\right)$, then the matching size of assignment of $\left(\begin{array}{ccc}
a& b &\\
a & &\\
a& b&c
\end{array}\right)$ is $2$. From $\left(\begin{array}{ccc}
\underline{a}&b &c\\
a & \underline{b} &\\
a& b&\underline{c}
\end{array}\right)$ and $\left(\begin{array}{ccc}
a& \underline{b} &c\\
\underline{a} &b &c\\
a& b&
\end{array}\right)$, we get $\left(\begin{array}{ccc}
a& \underline{b} &c\\
\underline{a} &b &\\
a& b&
\end{array}\right)$, then the matching size of assignment of $\left(\begin{array}{ccc}
a& b &c\\
a &b &\\
a& &
\end{array}\right)$ is $2$. From $\left(\begin{array}{ccc}
a& \underline{b} &c\\
\underline{a} &b &\\
a& b&
\end{array}\right)$ and $\left(\begin{array}{ccc}
\underline{a}& b &\\
a & \underline{b} &\\
a& b&\underline{c}
\end{array}\right)$, it follows that $\left(\begin{array}{ccc}
a & \underline{b} &\\
\underline{a}& b &\\
a& b&
\end{array}\right)$, we conclude $\left(\begin{array}{ccc}
\underline{a}& &\\
a & \underline{b} &\\
a& b&\underline{c}
\end{array}\right)$ is not true. Otherwise from $\left(\begin{array}{ccc}
\underline{a}& &\\
a & \underline{b} &\\
a& b&\underline{c}
\end{array}\right)$ and $\left(\begin{array}{ccc}
a & \underline{b} &\\
\underline{a}& b &\\
a& b&
\end{array}\right)$, it follows that $\left(\begin{array}{ccc}
a& &\\
\underline{a} & b &\\
a& b&
\end{array}\right)$, which contradicts to the Pareto optimality of the mechanism. Hence, the matching size of assignment of $\left(\begin{array}{ccc}
a& &\\
a &b &\\
a& b&c
\end{array}\right)$ is $2$. It is obvious to see that the matching size of assignment of $\left(\begin{array}{ccc}
a&b &c\\
a & &\\
a& b&
\end{array}\right)$ is at most $3$. Denote the current mechanism as $\phi^4$, we know that $\#^{\phi^4}(\mathcal{S})\le 13$.
\end{proof}
Note that Theorem \ref{lowerbound:n=3} shows that $\min_{I\in\mathcal{I}}\mathbb{E}_{p}(r(T_p,I))\le \max_{T\in\mathcal{T}}\mathbb{E}_{q}(r(T,S_{uni}))\le \frac{13}{18}$, for any $p\in \mathrm{P}$, and \bahar{when} $S_{uni}$ is the uniform distribution over $\mathcal{S}$. Hence, the approximation ratio is at least $\frac{18}{13}$.\\

\noindent
{\bf Lower bound for non-bossy mechanisms.} \piotr{In this subsection we only consider the unweighted HA problem and with strict preferences. Thus an instance of HA is just 
$I = (N,\objectset,L)$, where $L = (L(1),\ldots,L(n_1))$ is the joint list of (strict) preferences of the agents.}

We first define the concept of non-bossiness for a deterministic mechanism (see, e.g., \cite{Pa00}). A deterministic mechanism $\phi$ is non-bossy if 
for any agent $i \in \agentset$, any joint preference list profile $\preflist$, and any preference list $\preflist'(i)$ of agent $i$, if $\phi_i(\preflist(i),\preflist(-i)) = \phi_i(\preflist'(i),\preflist(-i))$ then $\phi(\preflist(i),\preflist(-i)) = \phi(\preflist'(i),\preflist(-i))$. Roughly speaking, non-bossiness ensures that no agent can change the allocation of other agents, by reporting a different preference list, without changing his own allocation.

\hide{
P\'apai \cite{Pa00} proved that non-bossiness \bahar{combined} with truthfulness is equivalent to \bahar{\emph{group-strategyproofness} (Lemma 1 in \cite{Pa00})} \todoB{I opted for using group-truthful instead of group-strategyproof since we use ``truthfulness'' to mean ``strategyproofness'' in our paper}. \bahar{A mechanism is group-strategyproof if no subset of agents can benefit by reporting false preferences.} We note here that the term strategyproof is equivalent to truthful, however we stick to the term group-strategyproof as is common in the literature.

\begin{lemma}[Lemma 1 in \cite{Pa00}]\label{l:papai} A mechanism for HA is group-strategyproof if and only if it is truthful and non-bossy.
\end{lemma}

Pycia and \"Unver \cite{PU17} 
\bahar{provided a characterization of group-truthful and Pareto optimal mechanisms.}
} 

Bade \cite{Ba18} showed that (Theorem 1 in \cite{Ba18}) any mechanism that is truthful, Pareto optimal and non-bossy is
\emph{\piotr{$s$}-equivalent} to SDM, in the sense that if the order of agents is generated uniformly at random,
the matching returned by SDM is the same as the one returned by any \piotr{truthful, Pareto optimal and non-bossy mechanism.}

\piotr{We will now briefly introduce the required notions from \cite{Ba18} to be able to formally use the result of Bade \cite{Ba18}. 
Let $\mechanism : \inputsettingset \rightarrow \matchingset$ be any deterministic mechanism for HA and $\sigma$ be any order (permutation) of the agents. We define a {\em permuted mechanism} $\sigma \odot \mechanism : \inputsettingset \rightarrow \matchingset$ via
$(\sigma \odot \mechanism )_i(L) = \mechanism_{\sigma^{-1}(i)}(L(\sigma(1)),\ldots,L(\sigma(n_1)))$ for any agent $i \in N$. 
Intuitively, permutation $\sigma$ assigns each agent in $N$ to a “role” in the mechanism, such that the agent $\sigma(i)$ 
under $\sigma \odot \mechanism$ assumes the role that agent $i$ plays under $\mechanism$.}

\piotr{The {\em symmetrization} of a deterministic mechanism $\mechanism : \inputsettingset \rightarrow \matchingset$ is a random mechanism
$Rand(\mechanism) : \inputsettingset \rightarrow Rand(\matchingset)$ that calculates the probability of matching
$\mu$ at the joint preferences list $L$ as the probability of a permutation $\sigma$ with $\mu = (\sigma \odot \mechanism)(L)$
under the uniform distribution on $\Pi(N)$, i.e., where $\sigma \in \Pi(N)$ is a uniform random permutation of the agents, and $\Pi(N)$ is the set of all permutations of the agents. So we have:}
$$
 \piotr{ \Pr{Rand(\mechanism)(L) = \mu} = \frac{|\{\sigma : (\sigma \odot \mechanism)(L) = \mu\}|}{n_1 !}.}
$$ \piotr{For instance a symmetrization of SDM is simply RSDM.}

\piotr{We say that two deterministic mechanisms $\phi$ and $\phi'$ are {\em s-equivalent} if $Rand(\phi) = Rand(\phi')$. The main result of Bade \cite{Ba18} can now be stated as.}

\piotr{
\begin{theorem}[Theorem 1 in \cite{Ba18}]\label{t:Bade} Any (deterministic) truthful, Pareto optimal and non-bossy mechanism for HA is s-equivalent to serial dictatorship mechanism (SDM).
\end{theorem}
} 

\hide{
Hence, by the arguments in the preliminaries that use Yao's minmax principle,  we have the following tight lower bound for any 
group-truthful
and Pareto optimal mechanism.
}

\piotr{Using this theorem we can now prove the following tight lower bound.}

\piotr{
\begin{theorem}\label{thm:grouplowerbound}
No random mechanism which is a symmetrization of any truthful, non-bossy and Pareto optimal mechanism can achieve the approximation ratio better than $\frac{e}{e-1}$.
\end{theorem}
}

\begin{proof}
 \piotr{Let $\mechanism$ be any deterministic truthful, non-bossy and Pareto optimal mechanism for the HA problem with strict preferences.
 By Theorem \ref{t:Bade} the random mechanism $Rand(\mechanism)$ is equivalent to RSDM. In the proof of Theorem \ref{thm:OSDM}, we have shown that the expected aproximation ratio of RSDM cannot be better than $\frac{e}{e-1}$ on the triangle instances of HA for large enough $n_1$. This concludes the argument.}
\end{proof}

\piotr{Note that our mechanism with strict preference lists and weights is a symmetrization of a truthful, non-bossy and Pareto optimal mechanism SDMT-2.}

\hide{
Note that our mechanism with strict preference lists and weights is universally truthful, non-bossy and Pareto optimal.
}

\section{Conclusion}
Whilst this paper has focused on Pareto optimality in the HA context, stronger forms of optimality are possible.  For example, \emph{minimum cost} (or \emph{maximum utility}), \emph{rank-maximal} and \emph{popular} matchings can also be studied in the HA context, and a matching of each of these types is Pareto optimal (see, e.g., \cite[Sec.\ 1.5]{Man13} for definitions).  As Pareto optimality is a unifying feature of all of these other forms of optimality, we chose to concentrate on this concept in our search for randomised truthful mechanisms that can provide good approximations to maximum matchings with desirable properties.
Note that the lower bound on the performance of deterministic truthful mechanisms that produce Pareto optimal matchigns extends to those producing matchings that satisfy these stronger optimality criteria.  
It will thus be the focus of future work to consider the performance of randomised truthful mechanisms for these problems.

\piotr{As far as lower bounds for randomised Pareto optimal mechanisms are concerned, we proved a lower bound of $\frac{18}{13}$ for any universally truthful Pareto optimal mechanism.  Moreover, we obtained a tight lower bound for the class of symmetrizations of truthful, Pareto optimal and non-bossy mechanisms using a characterization due to Bade \cite{Ba18}. We believe that the existence of a lower bound of $\frac{e}{e-1}$ for any universally truthful Pareto optimal mechanism is an interesting open question, and our lower bound of $\frac{18}{13}$ is a useful step towards resolving this.}

\section*{Acknowledgements}
\david{We like to convey our sincere gratitude to an anonymous reviewer for a very careful reading of this paper, and for valuable remarks that have helped to improve the presentation.  We would also like to thank anonymous reviewers of earlier versions of this paper for useful comments.  Finally, we would like to thank Anna Bogomolnaia and Herv\'e Moulin for helpful discussions concerning the results in this paper.}

\newpage
\section*{Appendix}
\setcounter{section}{3}
\setcounter{lemma}{7}
In this section we prove Theorem \ref{sdmt-1-any-PO} and also provide a systematic way of enumerating all Pareto optimal matchings. The following result is an important first step towards the former goal.
\begin{theorem}\label{thm:sp-po}
Any Pareto optimal matching is a strong priority matching for some ordering $\agentsorder$ of the agents. 
\end{theorem}
\begin{proof}
This proof makes use of the characterisation of Pareto optimal matchings in instances of HA (potentially with ties) given by Proposition \ref{prop:characterisation}.  Let a Pareto optimal matching $\matching$ be given.

Let $G=(V,E)$ be the \emph{envy graph} for $\matching$, defined as follows. In the graph, $V=\agentset$ and there is a directed edge from agent $i$ to agent $k$ if and only if $i$ weakly prefers $\matching(k)$ to $\matching(i)$. Every edge is colored. An edge $(i,k)$ is colored green if 
$\matching(k) \tie_{i} \matching(i)$, 
and is colored red otherwise---i.e., if 
$\matching(k) \strictlypref_{i} \matching(i)$. 
We claim that all the edges in every strongly connected component (SCC) of $G$ are green (we denote this claim by \textbf{C1}). To see this, note that by the definition of strongly connected components, there is a path from every node in a given component to every other node in the component. Hence, if there is a red edge $(i,k)$ in a SCC, then there must be a cycle with a red edge in the SCC (as there must be a path from $k$ to $i$). A cycle with at least one red edge corresponds to a cyclic coalition and hence $\matching$ could have not been a Pareto optimal matching, a contradiction.

Create graph $G'=(V',E')$ as follows. There is a vertex in $V'$ for each SCC of $G$, and there is a directed edge in $G'$ from $v_r$ to $v_s$, $v_r,v_s \in V'$ and $v_r \neq v_s$, if and only if there is an edge in $G$ from $i$ to $k$ for some $i$ and $k$ that belong to the SCCs of $v_r$ and $v_s$. It follows from the definition of strongly connected components that $G'$ is a DAG. Hence $G'$ admits a topological ordering. Let $X$ be a reversed topological ordering of $G'$. Let $\agentsorder=i_1, \ldots, i_{\numagents}$ be an ordering of all the agents that is consistent with $X$. 
That is, for every two agents $i_j$ and $i_r$, $1\leq j<r\leq \numagents$, the corresponding SCC of $i_j$ appears in $X$ no later than the corresponding SCC of $i_r$. (The order of the agents belonging to the same SCC can be determined arbitrarily.) We prove that $\matching$ is an SPM w.r.t.\ $\agentsorder$.

Assume, for a contradiction, that our claim does not hold. That is, $\matching$ is not an SPM w.r.t.\ $\agentsorder$. Hence there must exist another matching $\matching'$ which has a lexicographically smaller signature than $\matching$; i.e., $\sig(\matching') < \sig(\matching)$ (we denote this fact by \textbf{A1}). Let $i_j$ be the highest priority agent, w.r.t $\agentsorder$, such that $i_j$ \david{strictly} prefers his partner under $\matching'$ to his partner under $\matching$; i.e., $\matching'(i_j) \strictlypref_{i_j} \matching(i_j)$ (we denote this assumption by \textbf{A2}). Note that $\matching'(i_j)$ must be matched in $\matching$ or else $\matching$ admits an alternating coalition, namely $P=\left\langle i_j, \matching'(i_j) \right\rangle$---as $i_j$ \david{strictly} prefers unmatched object $\matching'(i_j)$ to his partner $\matching(i_j)$---and hence not Pareto optimal.
So $\matching'(i_j)$ is matched under $\matching$ to, say, $i_k$. Following \textbf{A2}, there must be a red edge from $i_j$ to $i_k$ in the envy graph $G$. Therefore, $i_k$ must have a higher priority than $i_j$ according to $\agentsorder$ (note that, by 
\textbf{C1}, $i_j$ and $i_k$ cannot belong to the same SCC).
It then follows from \textbf{A1} and \textbf{A2} that $i_k$ is matched under $\matching'$  and ranks his partners under $\matching$ and $\matching'$ the same; i.e.
$\matching(i_k) \tie_{i_k} \matching'(i_k)$. Now, $\matching'(i_k)$ must be matched in $\matching$ or else there is an alternating path coalition in $\matching$, namely $P=\left\langle i_j,i_k,\matching'(i_k)\right\rangle$, and hence $\matching$ is not Pareto optimal. Also, $i_k$ cannot be matched to $\matching(i_j$) or else there is a cyclic coalition in $\matching$, namely $P=\left\langle i_j,i_k\right\rangle$, and hence $\matching$ is not Pareto optimal. So $\matching'(i_k)$ is matched under $\matching$ to, say, $i_r$, $i_r \neq i_j$. Also, there is a green edge from $i_k$ to $i_r$ in $G$, since $\matching'(i_k)=\matching(i_r)$ and $\matching(i_k) \tie_{i_k} \matching'(i_k)$ . We claim that $i_r$ must have a higher priority than $i_j$ (we denote this claim by \textbf{C2}). To see this, first of all note that there cannot exist a path between $i_r$ and $i_j$ in $G$ or else $G$ admits a cycle with at least one red edge, namely $(i_j,i_k)$, and hence $\matching$ admits a cyclic coalition and is not Pareto optimal. Now, to complete the proof of \textbf{C2}, we consider two cases regarding whether there is a path from $i_r$ to $i_k$ or not.
\begin{itemize}
	\item Case 1: \emph{There is a path from $i_r$ to $i_k$ in $G$}.  Since there is an path from $i_k$ to $i_r$, then $i_r$ and $i_k$ must belong to the same strongly connected component. 
Therefore since $i_k$ has a higher priority than $i_j$ under $\agentsorder$, so must $i_r$. 
	\item Case 2: \emph{There is no path from $i_r$ to $i_k$ in $G$}. Therefore $i_r$ and $i_k$ belong to two different SCCs. Since there is an path from $i_k$ to $i_r$ in $G$, there is an edge in $G'$ from the SCC corresponding to $i_k$ to the SCC corresponding to $i_r$. Therefore $i_r$ must have a higher priority than $i_k$ under $X$ and thus under $\agentsorder$ as well. Hence, by transitivity, $i_r$ has a higher priority than $i_j$ under $\agentsorder$.
\end{itemize}
So far we have established that $i_r$ has a higher priority than $i_j$. Now, using a similar argument as for $i_k$, we can show that $i_r$ must be matched under $\matching'$ and must rank his partners under $\matching$ and $\matching'$ the same. Also, again using a similar argument as for $i_k$, $\matching'(i_r)$ is not the same object as $\matching(i_j)$ and $\matching'(i_r)$ must be matched in $\matching$ or else there is an alternating path coalition in $\matching$ contradicting the Pareto optimality of $\matching$. So $\matching'(i_r)$ is matched under $\matching$, to say $i_q$. Using exactly the same argument as we used for $i_r$ we can show that $i_q$ also has a higher priority than $i_j$. We can keep repeating this argument and every time we have to reach a new agent with a higher priority than $i_j$. However, there are a bounded number of agents and hence a bounded number of agents with higher priority than $i_j$, a contradiction.
\end{proof}

Using Theorem~\ref{thm:sp-po}, we show that SDMT-1 is capable of producing any given Pareto optimal matching.
\setcounter{lemma}{6}
\begin{theorem}
Any Pareto optimal matching can be generated by some execution of SDMT-1.
\end{theorem}
\begin{proof}
Let $\matching$ be a Pareto optimal matching for an instance $\inputsetting$ of HA. By Theorem~\ref{thm:sp-po}, $\matching$ is an SPM for some ordering $\agentsorder$ of the agents. Execute SDMT-1 given $\agentsorder$ as follows. At each phase $i$, choose $(i,\matching(i))$ as the augmenting path. Notice that since $\matching$ is a matching, both $i$ and $\matching(i)$ must be unmatched at the beginning of phase $i$. Futhermore, since $\matching$ is an SPM w.r.t $\agentsorder$, there cannot be an augmenting path from $i$ to an object that $i$ \david{strictly} prefers to $\matching(i)$.
\end{proof}

The above theorem implies that any Pareto optimal matching has a nonzero chance of materialising if we execute the following procedure: (1) randomly generate a priority ordering over the agents $\agentsorder$, (2) run SDMT-1 given $\agentsorder$, and (3) whenever faced with more than one choice, pick an augmenting path at random. Enumerating all Pareto optimal matchings in a more systematic way is however possible with the aid of Theorem~\ref{thm:sp-po} and the two forthcoming propositions, Proposition \ref{3.9} and Proposition \ref{3.10}. 

Following Theorem~\ref{thm:sp-po}, enumerating all Pareto optimal matchings is equivalent to enumerating all matchings $\matching$ such that $\matching$ is a strong priority \david{matching} w.r.t.\ some ordering of the agents. Recall that all SPMs w.r.t.\ $\agentsorder$ have the same signature. It hence follows that:

\setcounter{lemma}{8}
\begin{proposition}\label{3.9}
Let $\matching^*$ be a matching returned by SDMT-1 for a given priority ordering of the agents $\agentsorder=\agent_{1},\ldots,\agent_{\numagents}$. 
Then, a given matching $\matching$ is a strong priority matching w.r.t.\ $\agentsorder$ if and only if
\begin{itemize}
	\item the same set of agents are matched under both $\matching$ and $\matching^*$, and
	\item each matched agent $\agent$ is matched under $\matching$ to an object that he ranks the same as $\matching^*(\agent)$; i.e., $rank(\agent,\matching(\agent)) = rank(\agent,\matching^*(\agent))$.
\end{itemize}
\end{proposition}

Let $G(\agentsorder)=(V,E)$ be a graph where $V=\agentset^*\cup\machineset$ where $\agentset^*$ is the set of agents that are matched under $\matching^*$.  There is an edge between an agent $\agent\in\agentset^*$ and an object $\object\in\machineset$ if and only if $\agent$ ranks $\object$ the same as $\matching^*(\agent)$. It is then easy to see that:

\begin{proposition}\label{3.10}
A matching $\matching$ is a strong priority matching w.r.t.\ $\agentsorder$ if and only if it is a maximum cardinality matching in $G(\agentsorder)$.
\end{proposition}

\begin{theorem}
Given and instance $\inputsetting$ of HA, all Pareto optimal matchigns of $\inputsetting$ can be enumerated in time $O(\numagents!)$. 
\end{theorem}
\begin{proof}
It follows Theorem~\ref{thm:sp-po} that to enumerate all Pareto optimal matchings it is enough to execute the following procedure on all possible $\agentsorder$: (1) run SDMT-1 given $\agentsorder$, and (2) enumerate all SPMs w.r.t.\ $\agentsorder$. It follows Proposition \ref{3.10} that to enumerate all SPMs w.r.t.\ $\agentsorder$ we need only to enumerate all maximum cardinality matchings w.r.t.\ $G(\agentsorder)$. The latter can be achieved in $O(|V|)$ time per matching \cite{Uno97}.
\end{proof}

%
%
\end{document}